\documentclass[twocolumn]{IEEEtran}

\usepackage{amssymb, amsmath, amsfonts, amsthm}
\usepackage{caption, subcaption, cite}
\usepackage{graphicx, cite, mathtools}
\usepackage{dsfont}
\usepackage{enumitem}
\usepackage{braket}
\usepackage{xcolor}
\usepackage{comment}
\newtheorem{theorem}{Theorem}

\newtheorem{proposition}{Proposition}

\newcommand{\bl}[1]{\textcolor{blue}{#1}}
\allowdisplaybreaks

\begin{document}

\title{Generalized Selection in Wireless Powered Networks with Non-Linear Energy Harvesting}
\author{\IEEEauthorblockN{Maria Dimitropoulou, \IEEEmembership{Student Member, IEEE}, Constantinos Psomas, \IEEEmembership{Senior Member, IEEE}, and Ioannis Krikidis, \IEEEmembership{Fellow, IEEE}}
\thanks{This work was co-funded by the European Regional Development Fund and the Republic of Cyprus through the Research and Innovation Foundation, under the projects POST-DOC/0916/0256 (IMPULSE) and INFRASTRUCTURES/1216/0017 (IRIDA). It has also received funding from the European Research Council (ERC) under the European Union’s Horizon 2020 research and innovation programme (Grant agreement No. 819819).

M. Dimitropoulou, C. Psomas, and I. Krikidis are with the IRIDA Research Centre for Communication Technologies, Department of Electrical and Computer Engineering, University of Cyprus, Cyprus, e-mail:\{mdimit02, psomas, krikidis\}@ucy.ac.cy. Preliminary results of this work have been presented at the IEEE International Conference on Communications (ICC), 2020 \cite{ICC}.}} 

\maketitle

\begin{abstract}
The rapid growth of the so-called Internet of Things is expected to significantly expand and support the deployment of resource-limited devices. Therefore, intelligent scheduling protocols and technologies such as wireless power transfer, are important for the efficient implementation of these massive low-powered networks. This paper studies the performance of a wireless powered communication network, where multiple batteryless devices harvest radio-frequency from a dedicated transmitter in order to communicate with a common information receiver (IR). We investigate several novel selection schemes, corresponding to different channel state information requirements and implementation complexities. In particular, each scheme schedules the $k$-th best device based on: a) the end-to-end (e2e) signal-to-noise ratio (SNR), b) the energy harvested at the devices, c) the uplink transmission to the IR, and d) the conventional/legacy max-min selection policy. We consider a non-linear energy harvesting (EH) model and derive analytical expressions for the outage probability of each selection scheme by using tools from high order statistics. %Our results show that, the performance of all the proposed schemes converges to an error floor due to the saturation effects of the considered EH model.
Moreover, an asymptotic scenario in terms of the number of devices is considered and, by applying extreme value theory, the system's performance is evaluated. We derive a complete analytical framework that provides useful insights for the design and realization of such networks.
\end{abstract}

\begin{IEEEkeywords}
Extreme value theory, $k$-th best selection, order statistics, outage probability, wireless power transfer.
\end{IEEEkeywords}

\IEEEpeerreviewmaketitle

\section{Introduction}
The rapid evolution of the Internet of Things (IoT) leads inevitably to the large-scale deployment of low-powered devices and to a huge amount of information flow. A well-known and efficient way to deal with a large number of devices is the concept of scheduling \cite{order-statistics-}, i.e., selecting a single device to transmit. Indeed, by scheduling the user with the best channel, exploits channel fluctuations due to fading, and thus provides multiuser diversity and maximizes the sum capacity \cite{tse,goldsmith}. On the other hand, recharging these devices regularly becomes inconvenient, and in some cases even infeasible. To overcome this issue, wireless power transfer (WPT) is emerging as a promising solution for extending the lifetime of power-constrained devices by powering them remotely through dedicated radio-frequency (RF) signals \cite{WPCN1},\cite{survey}. Therefore, the proper utilization of scheduling protocols and WPT is necessary for the implementation of energy sustainable IoT networks via a reliable, efficient and controlled manner.

In order to enhance the performance of wireless communication systems, selection-diversity techniques which utilize the principle ``select the best'', are widely reported in the literature in the context of antenna selection \cite{alouinibook}, as well as relay selection in cooperative \cite{coop} and cognitive networks \cite{cognitive}. In the context of cooperative relaying networks, the work in \cite{best_relay} studies the outage performance of the best relay selection in adaptive decode-and-forward cooperative networks. Specifically, the best relay is selected as the relay node that can achieve the highest signal-to-noise ratio (SNR) at the destination node. Therefore, it is shown that the best relay selection reduces the amount of required resources, such as channels or time-slots, and achieves full diversity order. In cognitive radio networks, in order to improve the performance of secondary transmissions, an adaptive cooperation diversity scheme with a best relay selection policy is investigated in \cite{cognitive}. However, in practice, the best device may be unavailable due to some scheduling, load balancing or imperfect channel state information (CSI) feedback conditions \cite{why k-th}. Therefore, a more general selection diversity scheme that features the selection of the $k$-th best link constitutes an interesting practical solution. Furthermore, a generalized selection diversity scheme, improves the performance from a fairness standpoint but at the expense of diversity gain \cite{alouini3}. Generalized selection schemes have been considered in various communication scenarios such as cooperative relaying and cognitive radio networks, where the system's performance is analyzed, deriving closed-form expressions in terms of outage probability. The work in \cite{why k-th} analyzes the performance of adaptive decode-and-forward and amplify-and-forward cooperative diversity systems, with a generalized relay selection scheme over identical and non-identical Rayleigh fading channels. Furthermore, the performance analysis for underlay cognitive decode-and-forward relay networks with a generalized relay selection scheme is studied in \cite{why n-th} over independent and identically distributed (i.i.d.) Rayleigh fading channels. 

Since IoT networks employ a large number of devices, a practical tool in order to analyze the performance of such systems is extreme value theory (EVT), which provides tractable and accurate asymptotic expressions for several performance metrics \cite{EVT}. The work in \cite{evt} studies the average SNR and ergodic capacity of large-scale relay networks with best relay selection. By applying EVT, the authors obtain an implicit expression for the asymptotic cumulative distribution function (CDF) of the received SNR when the number of relays is high. In addition, an EVT approach for the asymptotic analysis of the effective average throughput and the average bit error probability of the $k$-th best link over different fading channels, is presented in \cite{alouini2}. Moreover, closed-form asymptotic expressions for the average and effective throughput of the $k$-th best secondary user in noise-limited and interference-limited secondary multiuser networks of underlay cognitive radio systems are derived in \cite{alouini3} and \cite{alouini4}, respectively.

In recent works, selection has also been considered in the context of wireless powered communications. In wireless communication systems, WPT can be realized by two main network architectures: wireless powered communication networks (WPCNs) and simultaneous wireless information and power transfer (SWIPT). WPCNs refers to the scenario, where WPT-based devices harvest energy broadcasted by a dedicated RF transmitter in order to power their uplink transmissions \cite{WPCN1}, \cite{WPCN2}. Specifically, orthogonal channels are allocated for energy and information transmission, respectively, which can be separated either in frequency or in time \cite{book_wpcn}. On the other hand, SWIPT refers to the scenario, where the transmitted RF signal is used simultaneously for both information and energy at the WPT-based devices \cite{SWIPT1}, \cite{SWIPT2}. In particular, performance benefits of relay selection for SWIPT are shown in \cite{diomidis}, where a trade-off between the information transfer to the designated receiver and the harvested energy to a set of energy harvesters is established. Furthermore, the work in \cite{morsi1} proposes multiuser scheduling schemes in which the trade-off between the sum rate and the average amount of harvested energy can be controlled. In order to exploit multiuser diversity and facilitate EH in a multiuser downlink SWIPT system, a joint user scheduling and power allocation algorithm is studied in \cite{morsi2} for the maximization of the long-term average total harvested power. Moreover, an adaptive scheduling scheme for EH-based multiuser systems is proposed in \cite{kim}, where multiuser diversity is investigated in terms of the harvested energy and achievable rate, providing a significant gain on both of these objectives at the same time, compared with conventional scheduling schemes. 
%The results can be easily extended to include an error rate analysis as well, a case which was omitted here due to space limitations.

The aforementioned works only consider selection in the context of SWIPT and assume that the EH process is linear. In contrast to previous works, in our paper, we examine a generalized selection problem in a WPCN taking into consideration a non-linear energy harvesting model. Specifically, the considered non-linear EH model captures practical limitations of the energy harvester, such as the saturation effect, i.e., the power level above which the output power remains constant \cite{clerckx}, \cite{bruno}. Moreover, we propose novel selection/scheduling schemes that feature a generalized selection, i.e., a $k$-th best selection based on a specific ordering. Finally, by utilizing EVT tools, we analyze the system's performance under an asymptotic scenario appropriate for IoT applications, where the number of the devices is large.
\begin{itemize}
\item We consider a network with batteryless devices that harvest RF energy from a dedicated energy transmitter (ET). A single device is selected for information transmission to a common information receiver (IR) based on a specific selection scheme. The selection mechanism consists of several novel selection/scheduling schemes corresponding to different implementation complexities and CSI requirements. Specifically, the proposed schemes are based on: a) the e2e SNR, b) the energy harvested at the devices, c) the uplink transmission to the IR, and d) the conventional max-min selection policy. 
\item By considering a generalized selection approach and by employing order statistics tools, we derive analytical expressions which characterize the outage probability of each of the proposed selection schemes. In particular, a complete analytical framework for the performance of the $k$-th best device is presented. We also extend this framework to take into account the joint selection of a pair of devices. Finally, by using tools from EVT, we evaluate the asymptotic performance of the system in terms of the number of the devices and a more tractable analytical framework is provided.
\item Our analysis considers a practical non-linear EH model, from which useful insights on the design of the network can be derived. Consequently, the system's performance converges to an error floor, regardless of the selection scheme, and closed-form expressions for the high SNR regime are derived. We show that each selection policy is ideal under specific conditions. However, in general, the selection scheme based on the e2e SNR provides the best performance among the proposed selection schemes. 
\end{itemize}

The rest of the paper is organized as follows. Section II presents the considered system model and the main assumptions. In Section III, we present the proposed selection schemes and analyze their outage probability performance. Then, Section IV provides asymptotic analytical results for the proposed selection schemes. Numerical results are provided in Section V and the paper is concluded with Section VI.

\emph{Notation}: $|z|$ denotes the magnitude of a complex variable $z$, $\mathbb P[X]$ denotes the probability of the event $X$, and $\mathbb E[X]$ represents the expected value of $X$. $K_1(\cdot)$ is the modified Bessel function of the second kind of the first order, $B\left(p,q\right)$ denotes the beta function \cite[8.38]{integrals}, $\Gamma(\cdot)$ denotes the gamma function \cite[8.31]{integrals}, $\gamma(p,q)$ denotes the incomplete gamma function \cite[8.35]{integrals}, $\log(\cdot)$ is the natural logarithm, and $\binom{M}{k}=\frac{M!}{(M-k)!k!}$ is the binomial coefficient.

\begin{table*}[t]\centering
\begin{tabular}{|l|l||l|l|}\hline
\bf{Notation} & \bf{Description} & \bf{Notation} & \bf{Description}\\\hline
$M$ & Number of devices & $P_t$ & Transmit power \\\hline
$g_i$ & Channel coefficient of $i$-th device between ET and $D_i$ &$E_i$ & Harvested energy of $i$-th device \\\hline
$h_i$ & Channel coefficient of $i$-th device between $D_i$ and IR & $\sigma_n^2$ & AWGN variance \\\hline
$a$, $b$, $c$ & Rectification circuit's parameters & $t_1T$ & harvesting phase duration \\\hline
$X_i$ & SNR of the $i$-th device & $t_2T$ & communication phase duration \\		\hline
$P_i$ & Transmission power of $i$-th device & $n_{\text{IR}}$ & IR AWGN \\		\hline
$\Pi^{(k)}_q$ & Outage probability of selection scheme $q$ & $\gamma_i$ & $i$-th ordered random variables \\\hline
$F_{\gamma_i}(x)$, $f_{\gamma_i}(x)$ & CDF, PDF of $\gamma_i$ & $i^*$, $j^*$ & $k$-th, $l$-th best device's index \\\hline
$Q$ & Information threshold & $\eta$, $\xi$ & Normalizing constants\\\hline
\end{tabular}
\caption{Summary of notations.}\label{tab:template}
\end{table*}

\section{System Model}\label{sys_model}
In this section, we describe the considered system model. The main mathematical notation related to the system model is summarized in Table I.

\subsection{Topology \& Channel Model}
We consider a WPCN topology consisting of an ET, $M$ i.i.d.\footnote{The i.i.d. channel assumption accommodates the selection process and it is a common approach in the literature \cite{coop}, \cite{why n-th}, \cite{alouini2}, \cite{alouini4}.} devices $D_i$, $i=1,\dots,M$, and a common IR; all the nodes are equipped with single antennas and the devices are located midway between the ET and IR. Time is slotted and the time slot duration is equal to $T$ (time units). During the harvesting phase with duration $t_1T$, $0 < t_1 < 1$, the ET transmits an RF energy signal with power $P_t$ to the devices, which harvest energy based on the received signal. During the communication phase with duration $t_2T$, where $t_2=1-t_1$, by applying the harvest-then-transmit protocol \cite{WPCN2}, the $k$-th best device (determined by the selection mechanism) transmits information to the IR by using all the harvested energy. We consider a Rayleigh block fading \footnote{This assumption is done for analytical tractability and other models are left for future consideration.} and therefore, the channel coefficients are complex Gaussian distributed with zero mean and unit variance. We denote by $g_i$ and $h_i$ the channel coefficient for the energy and information transmission, respectively, i.e., $g_i, h_i \sim CN(0,1)$. Finally, all wireless links in the network exhibit additive white Gaussian noise (AWGN) with variance $\sigma_n^2$. For the sake of simplicity, we assume a normalized slot duration $T=1$ (time units). Fig.1 illustrates the considered system topology.

\begin{figure}[t]\hspace{25mm} \centering
\includegraphics[width=8.75cm,height=7cm,keepaspectratio]{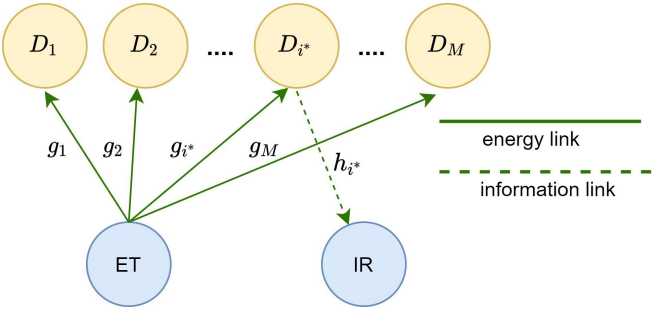}
\caption{A WPCN topology consisting of an ET that transmits power to $M$ batteryless devices and an IR that receives the information from the $k$-th best device.}
\end{figure}

\subsection{Energy Harvesting \& Information Transfer}
The devices of the network have WPT capabilities, and through a rectifying antenna (rectenna) convert RF to direct current power in order to operate their uplink transmission \cite{alouini1}. The rectification process is based on diode circuits and is a highly non-linear process. In order to capture this behavior, several practical EH models have been proposed {\cite{clerckx},} \cite{schober}, \cite{non linear}. In our analysis, we consider the non-linear EH model proposed in \cite{alouini1} that refers to specific excitation signals and is mathematically more tractable. Therefore, the energy $E_i$ harvested by the $i$-th device during the harvesting phase $t_1$ is described by

\begin{equation}
E_i = t_1\left(\frac{aP_t|g_i|^2+b}{P_t|g_i|^2+c}-\frac{b}{c}\right), \label{energy harvesting}
\end{equation}
where $|g_i|^2$ is the channel gain between the ET and the $i$-th device and $a$, $b$, $c$ are parameters determined by the rectification circuit through curve fitting \cite{alouini1}. Note that we also consider the harvested energy based on the linear EH model, which is described by $E_i = t_1P_t|g_i|^2$; the linear EH model is used as a useful performance benchmark across the paper.

During the communication phase, the received signal at the IR can be written as 
\begin{equation}
y_{\text{IR}}=h_i \sqrt{P_i}s_i+n_{\text{IR}},
\end{equation}
where $P_i=\frac{E_i}{1-t_1}$ is the transmission power available for the communication phase, $s_i$ is the $i$-th device's message and $n_{\text{IR}}\sim CN(0,\sigma_n^2)$ is the AWGN at the IR. Then, the SNR at the IR from the $i$-th device is given by 
\begin{equation}\label{snr}
X_i = \frac{|h_i|^2}{t_2\sigma^2_n} E_i,
\end{equation}
where $|h_i|^2$ is the channel gain between the $i$-th device and the IR. 

\subsection{Order-based Selection \& Extreme Value Theory}
Let $\gamma_i$ with $i \in \{1,\dots,M\}$ be $M$ i.i.d. random variables, corresponding to certain parameters which characterize the performance of the devices, that is, the SNR $X_i$, the energy harvested $E_i$, and the uplink channel gain $|h_i|^2$. Without loss of generality, we assume the following ordering \cite{alouinibook}
\begin{equation}\label{ordering}
\gamma_1 \leq \gamma_2 \leq \dots \leq \gamma_M,
\end{equation} 
where the knowledge of this ordering is acquired through a training period. % (available for selection).
Without loss of generality, we implement a proactive device selection, i.e., the selection is performed prior to the energy transmission. Now, since all the devices harvest energy and become active, the $k$-th best device that is selected, transmits its own data by using all the energy it harvested. The $k$-th best selection mechanism is based on the principles of each selection scheme and the details are described in the following section.

Assuming that $i^*$ denotes the $k$-th best device's index for each selection scheme, the PDF of $\gamma_{i^*}$ is given by \cite{alouini2}
\begin{equation}\label{pdf}
f_{\gamma_{i^*}}(x) = k \binom{M}{k} f_{\gamma_i}(x) F_{\gamma_i}(x)^{M-k} (1-F_{\gamma_i}(x))^{k-1},
\end{equation}
where $F_{\gamma_i}(x)$ and $f_{\gamma_i}(x)$ are the CDF and probability density function (PDF) of $\gamma_i$, respectively. According to the previous discussion, for scenarios where the selection is over a large number of devices, i.e., $M\rightarrow\infty$, we analyze the system performance by employing tools from EVT. Based on EVT, the distribution of $\gamma_M$, i.e., the maximum random variable, converges to one of three limiting distributions: the Gumbel distribution, the Fr\'echet distribution, or the Weibull distribution \cite{alouinibook}. Through this, the asymptotic distribution of the $k$-th best device can be approximated. %The type of the limiting distribution depends on the properties of the distribution functions of the original unordered random variables.% 
More details will be given in Section IV.

\section{Selection schemes and performance analysis}
In this section, we present the proposed selection schemes, each corresponding to a different implementation complexity and CSI requirements. Even though the devices will consume some energy during the training phase, or due to a circuit dissipation, we assume that this is negligible and that the energy consumption is dominated by the transmit power \cite{green radio}, \cite{consumption}. Their performance is analyzed in terms of the outage probability, defined as the probability that the information rate falls below a required threshold level. The general expression for the outage probability is given by 
\begin{equation}\label{random_op}
\Pi(x) = \mathbb P\{t_2\log_2(1+X_i)\leq Q\} = \mathbb P\{X_i\leq x\},
\end{equation}
where $X_i$ is the SNR of $i$-th device given by \eqref{snr}, $Q$ is the required threshold level and $x=2^{\frac{Q}{t_2}}-1$. We first present the random selection (RS) scheme, which is used as a useful performance benchmark and also assists the performance analysis of the proposed selection schemes. 

\subsection{Random Selection Scheme}
The RS scheme is a low implementation complexity scheme that does not require any CSI, where a device is randomly (uniformly) selected for information transmission. As all the devices harvest the energy signals from the ET, without loss of generality, we assume that the $i$-th device is selected to communicate with the IR. Then, the outage probability achieved by the RS scheme is given by the following theorem.
\begin{theorem}\label{theorem1}
The outage probability achieved by the RS scheme, is defined by 
\begin{align}\label{random cdf_alouini}
\Pi_{\rm{RS}}(x)=1-2 \sqrt{\frac{\sigma_n^2c^2t_2x}{P_tt_1(ac-b)}}\exp\left(-\frac{\sigma_n^2ct_2x}{t_1(ac-b)}\right)\nonumber\\
\times K_{1}\Bigg(2\sqrt{\frac{\sigma_n^2c^2t_2x}{P_tt_1(ac-b)}}\Bigg).
\end{align}	

\begin{proof}
See Appendix \ref{prf_theorem1}.
\end{proof}
\end{theorem}

The outage probability achieved by the RS scheme with the simplified linear EH model is
\begin{equation}\label{linear_random}
\Pi_{\rm{RS},\rm{L}}(x)=1-2\sqrt{\frac{\sigma_n^2t_2x}{P_t t_1}} K_1\left(2\sqrt{\frac{\sigma_n^2t_2x}{P_t t_1}}\right),
\end{equation}
where its derivation follows similar steps as the one for the non-linear model and so it is omitted. For the special case with $P_t\to\infty$, the asymptotic outage probability for the RS scheme with the non-linear model is given by 
	\begin{equation}\label{random}
\Pi^{\infty}_{\rm{RS}}(x)=1-\exp \left(-\frac{\sigma_n^2ct_2x}{t_1(ac-b)}\right),
\end{equation}
which follows the fact that $K_{\nu}(x) \approx \frac{1}{2} \Gamma(\nu) \big (\frac{x}{2} \big)^{-\nu},~ \nu>0$, as $x \to 0$ \cite{bessel_small}.
%\end{remark}
Observe that the asymptotic performance of the RS scheme is independent of the number of devices but depends on the rectenna's parameters and the ratio $\frac{t_2}{t_1}$. Also note that the linear EH model is not restricted by any saturation effects and thus the outage probability with $P_t\to\infty$ tends asymptotically to zero. This also holds for all the proposed $k$-th best selection schemes with the linear EH model.

As already mentioned, for many practical communication scenarios, the best device may not be available due to scheduling or load balancing conditions. Therefore, a more general strategy that features a generalized selection is of practical interest. %Specifically, although scheduling the $k$-th best device to access the channel for transmission may sacrifice some diversity gain, it provides significant performance gain from a fairness point of view.% 
In the following, we will focus on the impact of $k$-th best device selection on the outage probability for different selection schemes.
The general expression of the outage probability for the $k$-th best selection schemes in the considered system model is given by
\begin{align}
\Pi^{(k)}(x)=\mathbb P\{t_2\log_2(1+X_{i^*})\leq Q\}=\mathbb P\{X_{i^*}\leq x\},
\end{align}
where $X_{i^*}$ is the SNR of the $k$-th best device. Below the proposed $k$-th best device selection schemes are described in detail.

\subsection{SNR-based Selection (SBS) Scheme}
%\underline{\textit{Alouini's model}}
We first consider a selection scheme based on the e2e output SNR at the IR. Therefore, this scheme requires the knowledge of the received signal strength of all the links, i.e., both downlink and uplink. According to the SNR-based selection (SBS) scheme, the $k$-th best device is the one that achieves the $k$-th highest e2e output SNR. By assuming the ordering $X_{1}\leq X_{2}\leq \dots \leq X_{M}$, the index of the $k$-th best device selected for information transmission to the IR, is described by 
\begin{equation}
i^*=\arg\max^{(k)}_{i\in\{1,\dots,M\}}\{X_1,\dots,X_M\}.
\end{equation}

The outage performance, i.e., the CDF of the $k$-th best device's e2e output SNR, is described by the following theorem.
	\begin{theorem}\label{theorem2}
		The outage probability achieved by the SBS scheme where the $k$-th best device is selected, is defined by 
		\begin{equation}\label{op_sbs_alouini}
		\Pi^{(k)}_{\rm SBS}(x)=I_{F_{X_i}(x)} (M-k+1,k),
		\end{equation}
		where $I_{\psi}(p,q) = \frac{1}{B(p,q)} \int_0^\psi t^{p-1}(1-t)^{q-1}dt$ denotes the normalized incomplete beta function \cite[8.39]{integrals} and $F_{X_i}(x)$ describes the CDF of the $i$-th device's SNR given in \eqref{random cdf_alouini}.
			\end{theorem}
	\begin{proof}
	See Appendix \ref{prf_theorem2}.
\end{proof}
The outage probability for the SBS scheme with the linear EH model is given by \eqref{op_sbs_alouini} with $F_{X_i}(x)$ as in \eqref{linear_random}. For the special case where the best device is selected ($k=1$), the outage probability for the SBS scheme is reduced to
\begin{align}
\Pi^{(1)}_{\rm SBS}(x) = \Bigg(1-2 \sqrt{\frac{\sigma_n^2c^2t_2x}{P_tt_1(ac-b)}}\exp\left(-\frac{\sigma_n^2ct_2x}{t_1(ac-b)}\right) \nonumber \\
\times K_1\left(2\sqrt{\frac{\sigma_n^2c^2t_2x}{P_tt_1(ac-b)}}\right)\Bigg)^M,
\end{align}
for the non-linear EH model, and by 
\begin{align}
\Pi^{(1)}_{\rm SBS,L}(x) = \left(1-2\sqrt{\frac{\sigma_n^2t_2x}{P_t t_1}} K_1\left(\frac{2\sqrt{P_t t_1}\sqrt{\sigma_n^2t_2x}}{P_t t_1}\right)\right)^M,
\end{align}
for the linear EH model. Moreover, when $P_t\to\infty$, the asymptotic outage probability for the SBS scheme, is given by \eqref{op_sbs_alouini} with $F_{X_i}^\infty(x)$ described by \eqref{random}.

Similarly to the RS scheme, the asymptotic performance of the SBS scheme depends on the rectification parameters and the ratio $\frac{t_2}{t_1}$, i.e., the duration of the communication phase over the duration of the harvesting phase. However, it is clear that in contrast with the RS scheme, it also depends on the number of the devices as well as the parameter $k$, i.e., which device is selected. This remark provides useful insights, since for a specific number of devices, the rectenna's characteristics as well as the scheme's parameters can be designed accordingly to satisfy a certain outage requirement.

\subsection{Energy-based Selection (EBS) Scheme}
The energy-based selection (EBS) scheme makes a decision based solely on the achieved energy harvested at the devices. In other words, the EBS scheme selects the device that harvests the $k$-th most energy. Therefore, it is a low complexity scheme, and in contrast to the SBS scheme, requires only knowledge of the links' signal strength between the ET and $D_i$. By assuming the ordering $E_{1}\leq E_{2}\leq \dots \leq E_{M}$, the index of the $k$-th best device selected for information transmission to the IR, is
\begin{equation}
i^* = \arg\max^{(k)}_{i\in\{1,\dots,M\}}\{E_1,\dots,E_M\}.
\end{equation}
We provide the outage probability of this scheme in the following theorem.
\begin{theorem}\label{theorem3}
The outage performance of the $k$-th best device, achieved by the EBS scheme, is described by 
\begin{align}\label{op_ebs}
\Pi^{(k)}_{\rm EBS}(x) &= 1 - 2k \binom{M}{k} \sqrt{\frac{\sigma_n^2t_2x}{P_tt_1}} \nonumber \\
&\qquad \times \sum_{m=0}^{M-k}(-1)^m \binom{M-k}{m} \frac{\Phi(x,m)}{\sqrt{k+m}},
\end{align}
where
\begin{align}\label{nonlinear_ebs}
\Phi(x,m) &= \exp\left(-\frac{\sigma_n^2ct_2x}{t_1(ac-b)}\right) \sqrt{\frac{c}{ac-b}}\nonumber \\ &\qquad \times K_1\left(2\sqrt{\frac{\sigma_n^2c^2t_2x}{P_tt_1(ac-b)}(k+m)}\right).
\end{align}
\end{theorem}
	
\begin{proof}
See Appendix \ref{prf_theorem3}.
\end{proof}	
It is important to note here the impact of the harvesting/communication phase duaration, i.e., $\frac{t_2}{t_1}$ on the system's performance. From \eqref{op_ebs}, we observe that for the extreme case, i.e., $\frac{t_2}{t_1} \to 0$ the outage performance significantly degrades. In other words, when $\frac{t_2}{t_1} \to 0$, the harvesting phase becomes longer and dominates the communication phase, resulting in a limited time for the uplink transmission. This observation also holds for the rest of the selection schemes and is verified in Section V, where the trade-off between the outage performance and the time duration of the harvesting phase is also illustrated.

\noindent It is also worth mentioning that the selected device based on the EBS scheme may not harvest sufficient energy from the ET due to the saturation region of the considered EH model. However, the parameters $a,b,c$ which capture the rectenna's characteristics could be designed in such a way as to achieve the network's objective, i.e., a specific outage requirement.

When the simplified linear EH model is considered, the outage performance achieved by the EBS scheme is given by \eqref{op_ebs}, with 
\begin{equation}\label{linear_ebs}
\Phi_{\rm L}(x,m) = K_1\left(2\sqrt{\frac{\sigma_n^2t_2x}{P_tt_1}(k+m)}\right).
\end{equation}
If $k=1$, i.e., the best device is selected, the outage probability can be written as
\begin{align}
\Pi^{(1)}_{\rm EBS}(x) &= 1 - 2 M \exp\left(-\frac{\sigma_n^2ct_2x}{t_1(ac-b)}\right) \nonumber \\ 
& \times \sum_{m=0}^{M-1} (-1)^m \binom{M-1}{m} \sqrt{\frac{\sigma_n^2c^2t_2x}{P_tt_1(ac-b)(m+1)}}\nonumber\\
& \qquad \times K_1\left(2\sqrt{\frac{\sigma_n^2c^2t_2x}{P_tt_1(ac-b)}(m+1)}\right).
\end{align}
Similarly, for the linear EH model with $k=1$, we have
\begin{align}
\Pi^{(1)}_{\rm EBS,L}(x) &= 1 - 2 M \sum_{m=0}^{M-1}(-1)^m \binom{M-1}{m} \nonumber \\ &\qquad \times \sqrt{\frac{\sigma_n^2t_2x}{P_tt_1(m+1)}}K_{1}\!\left(2\sqrt{\frac{\sigma_n^2t_2x}{P_tt_1}(m+1)}\right)\!.
\end{align}

Also, by considering $P_t\to\infty$, the asymptotic outage probability based on the non-linear EH model, can be expressed as 
\begin{equation}
\Pi^{\infty}_{\rm EBS}(x) = 1 - \exp\left(-\frac{\sigma_n^2ct_2x}{t_1(ac-b)}\right).
\end{equation}
Observe that the asymptotic performance of the EBS scheme does not depend on the parameter $k$, i.e., which device is selected, since all the devices have the same EH performance at the high SNR regime. Recall that the RS scheme performs in the same way. Indeed, asymptotically when $P_t \to \infty$, the EBS and RS schemes converge to the same error floor, i.e., they become asymptotically equivalent (see \eqref{random}). Hence, due to this observation, it is preferable to employ the RS scheme at the high SNR regime, as it is of lower complexity. In order to overcome these limitations of the EBS scheme, we propose a selection scheme that ignores the harvesting phase and focuses entirely on the communication phase. 

\subsection{Information-based Selection (IBS) Scheme}
We now consider a selection scheme based on the information received at the IR, i.e., we prioritize the uplink transmission. Based on the information-based selection (IBS) scheme, the $k$-th best device is the one that has the $k$-th best uplink channel. In this case, in contrast to the EBS scheme, the knowledge required is only the links' signal strength between $D_i$ and IR. By assuming the ordering for the channels between $D_i$ and IR as $|h_1|^2\leq \ |h_2|^2\leq \dots \leq |h_M|^2$, the index of the $k$-th best device selected for communicating with the IR, is described by
\begin{equation} 
i^* = \arg\max^{(k)}_{i\in\{1,\dots,M\}}\{|h_1|^2,\dots,|h_M|^2\}.
\end{equation}

\begin{theorem}\label{theorem4}
The outage probability achieved by the IBS scheme with the $k$-th best device, can be written as	
\begin{equation}\label{op_ibs}
\Pi^{(k)}_{\rm IBS}(x) = 1 - k \binom{M}{k} \sum_{m=0}^{M-k} (-1)^m \binom{M-k}{m} \Phi(x,m),
\end{equation}
where
\begin{equation}
\Phi(x,m) = \int_{r}^{\infty} \exp\left(-(k+m)z-\frac{cr}{P_t(z-r)}\right)dz,
\end{equation}
with $r=\frac{\sigma_n^2ct_2x}{t_1(ac-b)}$.
\end{theorem}

\begin{proof}
See Appendix \ref{prf_theorem4}.
\end{proof}

Note that the outage probability achieved by the IBS scheme with the linear EH model is given by \eqref{op_ebs}, with $\Phi_L(x,m)$ described by \eqref{linear_ebs}. In other words, the linear EH model provides the same performance for both the IBS and EBS schemes. This holds due to the fact that the linear model is not restricted by any saturation effects and so this creates a symmetry between the downlink and uplink channels.

Now, for the special case where the best device is selected ($k=1$), the outage probability for the IBS scheme is given by
\begin{align}
\Pi^{(1)}_{\rm IBS}(x) &= 1 - 2 M \sum_{m=0}^{M-1} (-1)^m \binom{M-1}{m} \nonumber \\
&\qquad \times \int_{r}^{\infty} \exp\!\left(\!-(m\!+\!1)z\!-\!\frac{cr}{P_t(z-r)}\right)\!dz,
\end{align}
while for the linear EH model, is given by
\begin{align}
\Pi^{(1)}_{\rm IBS,L}(x) &= 1 - 2 M \sum_{m=0}^{M-1} (-1)^m \binom{M-1}{m} \nonumber \\
&\qquad \times \sqrt{\frac{\sigma_n^2t_2x}{t_1 P_t(m+1)}} K_1\left(2\sqrt{\frac{\sigma_n^2t_2x}{t_1 P_t}(m+1)}\right).
\end{align}
As before, we look at the case with $P_t\to\infty$. Then, the asymptotic outage probability for the EBS scheme is reduced to \begin{align}
\Pi^\infty_{\rm IBS}(x) &= 1 - k \binom{M}{k} \sum_{m=0}^{M-k} (-1)^m \binom{M-k}{m} \frac{1}{k+m}\nonumber \\ &\qquad \qquad \qquad \times\exp\left(-\frac{\sigma_n^2ct_2x}{t_1(ac-b)}(k+m)\right).
\end{align} 
It is clear that in the high SNR regime and in contrast to the EBS scheme, the IBS scheme depends on the parameter $k$. Indeed, it is easy to see that for $k=1$, we have $\Pi^{\infty}_{\rm IBS}(x)<\Pi^{\infty}_{\rm EBS}(x)$. This justifies the consideration of this scheme in order to overcome the limitations of the EBS scheme. However, this inequality does not always hold. In particular, for $k=M$, i.e., selecting the last device, it can be easily deduced that $\Pi^{\infty}_{\rm IBS}(x)>\Pi^{\infty}_{\rm EBS}(x)$. It is also important to note that for high SNR regime, IBS converges to the SBS scheme. This can be easily deduced by considering $P_t \to \infty$ in (58) and (67), respectively, where we can obtain $\Pi^\infty_{\rm SBS}(x) =\Pi^\infty_{\rm IBS}(x)$. This observation is reasonable because with $P_t\rightarrow\infty$, all the devices are in the saturation region, and therefore, the uplink transmission dominates the harvesting phase.

\subsection{Max-min Selection (MMS) Scheme}
Finally, we look at a well-known selection scheme in the literature, namely, the max-min selection (MMS) \cite{maxmin1}, which is suitable for distributed implementation. The MMS constitutes a diversity-optimal strategy for relay selection and user scheduling and has been extensively considered in cooperative networks, due to its efficiency and low implementation complexity. As with SBS scheme, this selection scheme requires the knowledge of all the links. According to the MMS scheme, the worst link between the uplink and downlink of each device is determined and then the device with the $k$-th strongest worst link is selected. By assuming that the worst link of each device is denoted by $\rho_i=\min\{|g_i|^2, |h_i|^2\}$, ${i\in\{1,\dots,M\}}$ and by considering the ordering $\rho_1\leq \rho_2\leq \dots \leq \rho_M$, the index of the $k$-th best device is
\begin{equation}
i^* = \arg\max^{(k)}_{i\in\{1,\dots,M\}}\{\rho_1,\dots,\rho_M\}.
\end{equation}

\begin{theorem}\label{theorem5}
The outage performance of the $k$-th best device selection based on the MMS scheme is expressed as 
\begin{align}
\Pi^{(k)}_{\rm MMS}(x) &= k \binom{M}{k} \sum_{m=0}^{M-k} (-1)^m \binom{M-k}{m}\bigg(\frac{1-\exp(-2\delta s)}{\delta}\nonumber\\
&\qquad \qquad -\int_0^s\exp\left(-w-(2\delta-1)y\right)dy\nonumber \\
&\qquad \qquad -\int_r^s\exp\left(-v-(2\delta-1)z\right)dz\bigg),\label{op_mms}
\end{align}
where $v=-\frac{cr}{P_t(r-z)}$, $w=r+\frac{cr}{P_ty}$, $s=\sqrt{\frac{r^2}{4}+\frac{cr}{P_t}}+\frac{r}{2}$, $\delta=k+m$, and $r$ is given in Theorem  \ref{theorem4}.
\end{theorem}

\begin{proof}
See Appendix \ref{prf_theorem5}.
\end{proof}

The MMS scheme with linear EH model achieves the following outage probability
\begin{align}
\!\!\!\Pi^{(k)}_{\rm MMS,L}(x) &= k \binom{M}{k}\!\! \sum_{m=0}^{M-k} (-1)^m \binom{M-k}{m}\! \nonumber \\
&\times \bigg(\!\frac{1\!-\!\exp(-2\delta s)}{\delta}\!-\!\frac{2\left(\exp(-s)\!-\!\exp(-2\delta s)\right)}{2\delta-1}\!\bigg),\label{linear_mms}
\end{align}
where $v=w=s=\sqrt{\frac{\sigma_n^2t_2x}{t_1 P_t}}$, $\delta=k+m$. When the best device is selected ($k=1$), the outage probability for the MMS scheme is given by
\begin{align}
\Pi^{(1)}_{\rm MMS}(x) &\!=\! M\! \sum_{m=0}^{M-1} (-1)^m \binom{M-1}{m} \bigg(\frac{1-\exp(-2(m+1)s)}{m+1}\nonumber\\
&-\exp\left(-r\right) \int_0^s \exp\left(-\frac{cr}{y}	
-(2m+1)y\right)dy\nonumber\\
&-\int_r^s \exp\left(v-(2m+1)z\right)dz\bigg).
\end{align}

For the linear EH model, we have
\begin{align}
\Pi^{(1)}_{\rm MMS,L}(x) &= \!\!M\!\!\sum_{m=0}^{M-1} \!(-1)^m \binom{M-1}{m}\! \bigg(\frac{\!1-\!\exp(-2(m+1) s)}{m+1}\nonumber \\
&\qquad-\frac{2(\exp(-s)-\exp(-(2m+2)s))}{2m+1}\bigg).
\end{align}

Finally, for $P_t\to\infty$, the asymptotic outage probability for the MMS scheme is 
\begin{align}
\Pi^\infty_{\rm MMS}(x) &= k \binom{M}{k} \sum_{m=0}^{M-k}(-1)^m \binom{M-k}{m}\Bigg(\frac{1-\exp\left(-2\delta r\right)}{\delta}\nonumber\\
&\qquad \qquad-\frac{\exp\left(-r\right)-\exp\left(-2(\delta-1)r\right)}{2\delta-1}\Bigg),
\end{align}
which, as expected, depends on the number of the devices, the $k$-th selected device and the parameters of the rectenna, as the SBS scheme. 

The previous analysis focuses on the selection of a single device. Although the single device selection is of practical interest, due to the massive connectivity, more than one device might be selected for information transmission. For the sake of simplicity, the selection of a pair of devices is considered. In the next subsection, we extend our analysis for scenarios where two devices are selected and simultaneously access the channel. 

\subsection{Pair Device Selection}
Here, we examine the impact of the joint selection of two devices, say the $k$-th and $l$-th best devices, on the performance of the system. By considering a single user detection at the IR, the general expression for the outage probability is given by
\begin{align}
\Pi^{(k,j)}(x)\!&=\!\mathbb P\!\left\{t_2\log_2(1\!+\!X^{(k)})\!\leq \! Q \cap t_2\log_2(1\!+\!X^{(j)})\!\leq \!Q\!\right\}\nonumber\\
&=\mathbb P \left\{\frac{X_{i^*}}{X_{j^*}+\sigma_n^2}\leq x \cap\frac{X_{j^*}}{X_{i^*}+\sigma_n^2}\leq x \right\}, \label{joint}
\end{align}
where $X_{i^*}$ and $X_{j^*}$ are the SNRs of the $k$-th and $l$-th best device, respectively.
Other more sophisticated multiple access techniques can also be used but the single user detection keeps the complexity low and is used for simplicity. For the sake of brevity, we focus on the RS and SBS schemes but the analytical methodology follows a similar approach for the other schemes as well.

According to the SNR-based pair device selection, the $k$-th and $l$-th best device are the ones that achieve the $k$-th and $l$-th highest SNR. By considering the ordering $X_{1}\leq X_{2}\leq \dots \leq X_{M}$, the indices for $k$-th and $l$-th best device are given by
\begin{equation}
i^*=\arg\max^{(k)}_{i\in\{1,\dots,M\}}\{X_1,\dots,X_M\},
\end{equation}
and
\begin{equation}
j^*=\arg\max^{(j)}_{i\in\{1,\dots,M\}}\{X_1,\dots,X_M\},
\end{equation}
respectively. Then, by assuming $i^*<j^*$ and $z<y$, the joint PDF of the $k$-th and $l$-th best devices' SNR is described by \cite{order-statistics-}
\begin{equation}\label{joint_pdf}
p_{X_{i^*}\!,X_{j^*}}\!(y,\!z)\!\! =\!\! \frac{\!M! (1\!-\!F_{X_i}\!(y))^{j-1}\!\! f_{X_i}\!(y)\! (F_{X_i}\!(y)\!\!-\!\!F_{X_i}\!(z))^{k\!-\!j\!-\!1}\!}{\!(j\!-\!1)!(k\!-\!j\!-\!1)!(M\!\!-\!\!k\!)!f_{X_i}\!(y)\!F_{X_i}\!(z)^{M\!-\!k}\!},
\end{equation}
where the CDF of the $i$-th device's SNR is given by \eqref{random cdf_alouini}. 

By combining \eqref{joint_pdf} and \eqref{joint}, we derive the joint outage performance, i.e., the joint CDF of the $k$-th and $l$-th best device's SNR, provided in the following theorem.

\begin{theorem}\label{theorem6}
The joint outage performance of the $k$-th and $l$-th best device achieved by the SBS scheme can be expressed as
\begin{equation}\label{theorem6_op}
\Pi^{(k,j)}_{\rm SBS}(x)=\int_{0}^{\frac{\sigma_n^2x}{1-x}} \int_{\max\{z,\frac{z-\sigma_n^2x}{x}\}}^{x (z+\sigma_n^2)} p_{X_{i^*},X_{j^*}}(y,z) dy dz,
\end{equation}
where $p_{X_{i^*},X_{j^*}}(y,z)$ is given by \eqref{joint_pdf}.
\end{theorem}
It is important to point out here, that there is a constraint on the parameter $x$. Specifically, since $i^*<j^*$, the inequality $\frac{X_{i^*}}{X_{j^*}+\sigma_n^2}<1$ must hold. Hence, we obtain that $x<1$ which is in line with the analytical constraint of the integral's lower limit in \eqref{theorem6_op}, i.e., $1-x>0$.
The joint outage probability based on the linear EH model is described by \eqref{theorem6_op}, where the joint pdf is given by \eqref{joint_pdf} with CDF of the $i$-th device's SNR as in \eqref{linear_random}. When $P_t\to\infty$, the asymptotic joint outage probability can also be written as \eqref{theorem6_op}, where the joint PDF is described by \eqref{joint_pdf} with the CDF of the $i$-th device's SNR as in \eqref{random}.
By considering the RS scheme, the joint outage probability can be written as
\begin{equation}\label{theorem6_op_rand}
\Pi^{(k,j)}_{\rm RS}(x)=\int_{0}^{\frac{\sigma_n^2x}{1-x}} \int_{\max\{0,\frac{z-\sigma_n^2x}{x}\}}^{x (z+\sigma_n^2)} p_{X_{i^*},X_{j^*}}(y,z) dy dz,
\end{equation}
with joint PDF described by
\begin{equation}\label{joint_rs}
p_{X_{i^*},X_{j^*}}(y,z) =f_{X_i}(y) f_{X_i}(z),
\end{equation}
where the PDF of the $i$-th device's SNR is given by \eqref{random cdf_alouini}. 
For the linear EH model, following the same analytical steps with SBS scheme, the outage probability can be expressed as \eqref{theorem6_op} with the joint pdf described by \eqref{joint_rs}. Similarly with the SBS scheme, when $P_t\to\infty$, the asymptotic joint outage probability can be also written as \eqref{theorem6_op}, where the joint pdf is given by \eqref{joint_rs} with the CDF of the $i$-th device's SNR as in \eqref{random}.

In what follows, we will consider a more pracical scenario in IoT, where the number of the devices significantly increases. The impact of the number of the devices on the system performance is analyzed extensively by using EVT tools.

\section{EVT-based performance analysis}
In this section, we examine the asymptotic performance of the system in terms of the number of the devices, i.e., when $M \to \infty$. We analyze the performance of the proposed $k$-th best selection schemes by deriving the limiting distribution of the $k$-th best device and evaluate the asymptotic outage probability based on EVT tools. Through this methodology, we observe that the analytical expressions for the outage probability are simplified compared to the finite case, which we derived in the previous section. 

Consider the ordering given by \eqref{ordering}, where $\gamma_M$ denotes the largest order statistic of $M$ i.i.d. random variables. Furthermore, assume that $\frac{\gamma_M-\eta}{\xi}$ has a limiting CDF $G(x)$, where $\eta$ and $\xi$ are normalizing constants. Then, for a fixed $k$ and $M \to \infty$, the limiting CDF of $\frac{\gamma_{M-k+1}-\eta}{\xi}$ can be written as \cite{order-statistics-} 
\begin{equation}\label{cdf_lim}
G^{(k)}(x)=G(x) \sum_{j=0}^{k-1} \frac{(-\log(G(x)))^{j}}{j!}.
\end{equation}

Recall that the limiting distribution of the maximum of i.i.d. random variables converges to one of three limiting distributions; the Gumbel distribution, the Fr\'echet distribution, or the Weibull distribution.
In our case, we can easily prove that the following criterion
\begin{equation}
\lim_{x\to\infty} \frac{1-F_{\gamma_i}(x)}{f_{\gamma_i}(x)}=\lambda, \text{ } \lambda>0,
\end{equation}
is valid for all the selection schemes, where $F_{\gamma_i}(x)$ and $f_{\gamma_i}(x)$ are the CDF and PDF of $\gamma_i$, respectively. Therefore, it follows that the limiting distribution of the best device is of the Gumbel type with CDF given by
\begin{equation}
G(x)=\exp(-\exp(-x)), ~~-\infty < x < \infty.
\end{equation}
The normalizing constants $\xi>0$ and $\eta$ satisfy the following condition 
$\lim_{M\to\infty}F_{\gamma_M}(\xi x+\eta)=G(x)$,
where $F_{\gamma_M}(\cdot)$ is the CDF of the best device. These constants can be obtained by solving the following equations \cite{alouinibook}
\begin{equation} \label{constant_eta}
1-F_{\gamma_i}(\eta)=\frac{1}{M}, 
\end{equation}
and
\begin{equation}\label{constant_ksi}
1-F_{\gamma_i}(\eta+\xi)=\frac{1}{eM}, 
\end{equation}
where $e$ is Euler's number. The CDF of the limiting distribution of the $k$-th best device for fixed $k$ and $M\to\infty$ is described by \cite{alouini2} 
\begin{equation}
G^{(k)}(x)=\exp(-\exp(-x)) \sum_{j=0}^{k-1} \frac{\exp(-jx)}{j!}.
\end{equation}
Therefore, the outage probability of the $k$-th best can be approximated as \cite{alouini2} 
\begin{align}\label{op_evt}
\Pi^{(k)}(x)& \triangleq \mathbb P\{X_{M-k+1} \leq x\} \nonumber \\
&= \mathbb P\left\{\frac{X_{M-k+1}-\eta}{\xi} \leq \frac{x-\eta}{\xi}\right\} \nonumber\\
& \approx \mathbb P\left\{Z \leq \frac{x-\eta}{\xi}\right\}=G^{(k)}\left(\frac{x-\eta}{\xi}\right).
\end{align}

\subsection{Single Device Selection}
In the following propositions, we derive the asymptotic analytical expressions in terms of the number of the devices for the outage probability achieved by the proposed $k$-th best selection schemes.
\begin{proposition}\label{proposition_1} The asymptotic outage probability of the $k$-th best device achieved by the SBS scheme can be written as
\begin{align} \label{evt_sbs}
\Pi^{(k)}_{\rm SBS}(x) &= \exp\left(-\exp\left(-\frac{x-\eta^{\rm SBS}}{\xi^{\rm SBS}}\right)\right) \nonumber \\
&\qquad \times \sum_{j=0}^{k-1} \frac{1}{j!}\exp\left(-j\frac{x-\eta^{\rm SBS}}{\xi^{\rm SBS}}\right).
\end{align}
\end{proposition}

%\begin{proof}
%See Appendix \ref{prf_proposition1}.
%\end{proof}

By substituting the CDF of the $i$-th device's SNR described by \eqref{random cdf_alouini}, in \eqref{constant_eta}, \eqref{constant_ksi}, we evaluate numerically the normalizing constants $\eta^{\rm SBS}=F_{X_i}^{-1}(1-\frac{1}{M})$ and $\xi^{\rm SBS}=F_{X_i}^{-1}(1-\frac{1}{eM})-\eta^{\rm SBS}$. 
%For the linear EH model, the asymptotic outage probability can be expressed as in \eqref{evt_sbs} by substituting the normalizing constants $\eta$ and $\xi$, which are obtained based on the CDF of the $i$-th device's SNR in $\eqref{linear_random}$.

%\bl{The asymptotic outage performance described by \eqref{evt_sbs} is in a generalized form since can be applied for any CDF.}
\begin{proposition}\label{proposition_2} The asymptotic outage performance of the $k$-th best device based on EBS scheme can be expressed as follows

\begin{align}\label{ebs_evt}
	\Pi^{(k)}_{\rm EBS}(x)&\!=\!1-\frac{1}{\Gamma(k)}\exp(-r)\sum_{n=0}^{\infty}\frac{(-1)^n M^{n+k}}{n!}\sqrt{\frac{4rc}{P_t(n+k)}}\nonumber \\
	&\qquad \qquad \qquad \qquad \times K_1\left(\frac{4rc(n+k)}{P_t}\right). 
	\end{align}

where $r$ is given in Theorem  \ref{theorem4}.
\end{proposition}

\begin{proof}
See Appendix \ref{prf_proposition2}.
\end{proof}

%For the linear EH model, we have
%\begin{align}\label{ebs_evt_linear}
%\Pi^{(k)}_{\rm EBS}(x)=1-\frac{M^k}{\Gamma(k)}\int_0^\infty\exp\left(-M\exp(-y)-ky-\frac{\sigma_n^2t_2x}{P_tt_1y}\right)dy. 
%\end{align}

\begin{proposition} \label{proposition_3} The asymptotic outage probability of the $k$-th best device for the IBS scheme is given by	
\begin{align}
		\Pi^{(k)}_{\rm IBS}(x) &\!= \!\frac{1}{\Gamma(k)}\sum_{n=0}^{\infty}M^{n+k} \frac{(-1)^n}{n!}\bigg(\frac{\exp(-r(n+k))}{n+k} \nonumber \\ &\qquad-\sum_{m=0}^{\infty}\frac{1}{m!}\left(-\frac{cr}{P_t}\right)^m \!\!\int_{r}^{\infty}\!\!\frac{\exp(-z(n+k))}{(z-r)^{-m}}dz\bigg),
\end{align}
where $r$ is given in Theorem \ref{theorem4}.

\end{proposition}

\begin{proof}
See Appendix \ref{prf_proposition3}.
\end{proof}

%For the linear EH model, the asymptotic outage probability is given by \eqref{ebs_evt_linear} following the same principles with the finite case, as before.
\begin{proposition} \label{proposition_4}
The asymptotic outage performance of the $k$-th best selection achieved by the MMS scheme is described by
\begin{align} \label{mms_asymptotic}
\Pi^{(k)}_{\rm MMS}(x)&=\frac{1}{2}\frac{M^{k+\zeta}}{\Gamma(k)}\sum_{n=0}^{\infty}\frac{1}{n!}\sum_{m=0}^{\infty}\!\binom{n}{m}\!(-1)^{\zeta} (-\zeta)\!^{-m-1}\! \nonumber \\
&\times \bigg(\!\!\left(\!k+\frac{1}{2}\right)^m\!\!\!\!\theta(r)\!-\!k^m\!\exp(-v)\bigg(\theta(r)\!-\!\theta(s)\bigg)\!\!\bigg)\nonumber\\
&\quad \nonumber \\ 
&\quad+\frac{k}{2}\binom{M}{k}\sum_{m=0}^{M-k}(-1)^m\frac{1-\exp(-2(k+m)r)}{k+m},
\end{align}
where $\zeta=2(n-m)$, $\theta(\omega)=\gamma(m+1,2 \zeta\omega)$ and $r$ and $s$ are given in Theorems \ref{theorem4} and \ref{theorem5}, respectively.
\end{proposition}

\begin{proof}
See Appendix \ref{prf_proposition4}.
\end{proof}

The system's parameters apart from $M$, affect the performance in the same way as in the finite case, which has been discussed in Section III. What is important here is how $M$ affects the outage performance. Regarding Proposition 1, we observe that $M$ is incorporated in the argument of the term $\exp\left(-\exp\left(-\frac{x-\eta^{\rm SBS}}{\xi^{\rm SBS}}\right)\right)$ and, therefore, an increase in $M$ results in an enhanced outage performance. Moreover, as the parameter $k$ increases, the outage probability in \eqref{evt_sbs} decreases slowly. On the other hand, the EBS and the IBS schemes are affected from the factor $M^{n+k}$. Therefore, it can be noted that a decrease in the outage performance is slower compared to the SBS scheme. This observation is verified in Section V.
%For the linear EH model, we have
%\begin{equation} \label{linear_mms_evt}
%\Pi^{(k)}_{\rm %MMS}(x)=\frac{1}{2}\frac{M^k}{\Gamma(k)}\sum_{n=0}^{\infty}\frac{1}{n!}\sum_{m=0}^{\infty}\binom{n}{m}(-M)^{\zeta}(-\zeta)^{-m-1}%\left(k+\frac{1}{2}\right)^m\theta(s).
%\end{equation}
%\bl{Observe that compared to Theorem \ref{theorem2} and Theorem \ref{theorem3}, Proposition \ref{proposition_1} and Proposition \ref{proposition_2}, respectively, due to their simplest mathematical form, can provide further system insights by assisting in . Specifically, it is easier to obtain direct insights about the system's parameters in order to be designed accordingly, since the expressions are simplified. 

Similar to the finite case, we examine the pair device selection, i.e., the joint selection of $k$-th and $l$-th best device, based on the SBS scheme. 
\subsection{Pair Device Selection}
Assuming that $j^*$ denotes the $l$-th best device's index, the output SINRs of the $k$-th and $l$-th best device are given by
\begin{equation}
X^{(k)}=\frac{X_{i^*}}{X_{j^*}+\sigma_n^2},~~ 
X^{(j)}=\frac{X_{j^*}}{X_{i^*}+\sigma_n^2}.
\end{equation}

In order to analyze the asymptotic behavior of the system, where $M\rightarrow\infty$, we need to obtain the asymptotic joint distribution. We know that the difference $\mathbb P \{\gamma_1 \leq x_1, \gamma_M \leq x_M\}-\mathbb P \{\gamma_1 \leq x_1\} \mathbb P \{\gamma_M \leq x_M\}$ tends to zero for $M \to \infty$, whatever the underlying distribution is \cite{walsh}. In other words, any ``lower" extreme is asymptotically independent of any ``upper" extreme \cite{order-statistics-}. Due to this observation, the asymptotic joint outage probability of the $k$-th and $l$-th best can be approximated by
\begin{align}
\Pi^{(k,j)}(x)&=\!\mathbb P\{t_2\log_2(1\!+\!X^{(k)})\!\leq \!Q \cap t_2\log_2(1\!+\!X^{(j)})\!\leq \!Q\}\nonumber\\
&=\mathbb P \bigg \{\frac{X_{i^*}}{X_{j^*}+\sigma_n^2}\leq x \cap\frac{X_{j^*}}{X_{i^*}+\sigma_n^2}\leq x \bigg \}\nonumber\\
&=F_{X^{(k)}}(x) F_{X^{(j)}}(x),
\end{align}
where $F_{X^{(k)}}(x)$ and $F_{X^{(j)}}(x)$ are the marginal CDFs of the $k$-th and $l$-th best device's SNR, respectively, described by
\begin{equation}
F_{X^{(k)}}(x)=\int_{0}^{\infty} \int_{z}^{x (z+\sigma_n^2)} p_{X_j,X_k}(y,z) dy dz,
\end{equation}
and 
\begin{equation}
F_{X^{(j)}}(x)=\int_{0}^{\infty} \int_{0}^{\min\{x(z+\sigma_n^2), z\}} p_{X_j,X_k}(y,z) dy dz,
\end{equation}
where $p_{X_j,X_k}(y,z)$ is the joint PDF given by \eqref{joint_pdf}.

\section{Numerical Results}
In this section, we validate the derived analytical expressions with Monte Carlo simulations. % by using $10^6$ iterations. 
The analytical results are illustrated with lines and the simulation results with markers. The selection of the $k$-th best device occurs from a set of $M=5$ devices. The normalizing constants for the non-linear EH model are set as $a=2.463$, $b=1.635$, $c = 0.826$ \cite{alouini1}. Moreover, %The harvesting phase is set with duration $t_1=0.5256$ (time units) which is the optimal for the outage probability, 
the transmit power is set as $P_t=-10$ dBm, the required threshold level as $Q=0$ dB and the AWGN with variance $\sigma_n^2=-50$ dBm \cite{dbm}. The time duration of the harvesting phase is set as $t_1=0.5$. Note that these parameters have been chosen for the sake of presentation and a different selection of these values will lead to the same observations.

\begin{figure}[t]
\begin{subfigure}{\linewidth}\centering
		\includegraphics[width=\linewidth]{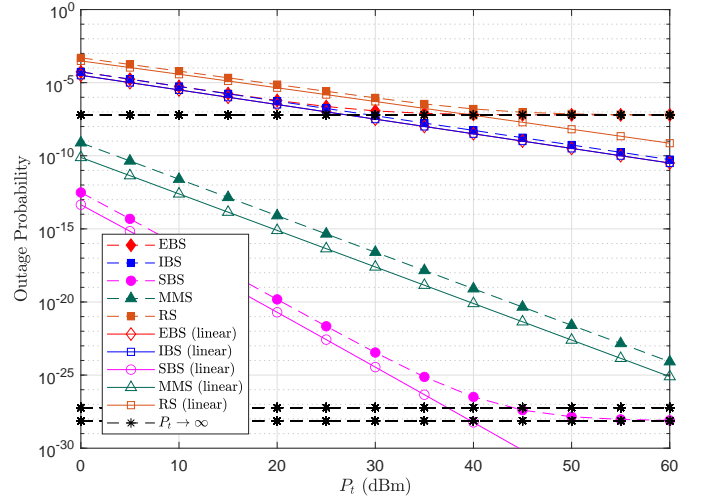}
		\caption{$k=2$.}
		\label{fig_2a}
	\end{subfigure}	\hfill
	\begin{subfigure}{\linewidth}\centering
		\includegraphics[width=\linewidth]{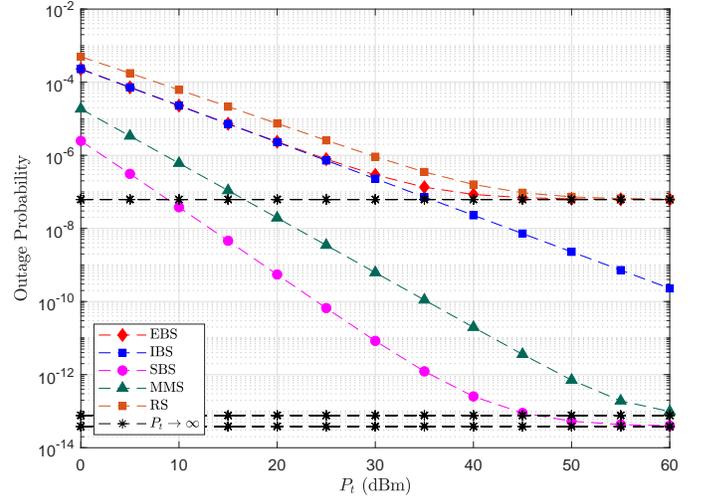}
		\caption{$k=4$.}
		\label{fig_2b}
	\end{subfigure}	\hfill
	\caption{Outage probability versus transmit power $P_t$; $M=5$.}
\end{figure}

Fig. \ref{fig_2a} illustrates the outage probability performance of the proposed selection schemes in terms of the transmit power $P_t$. As expected, the performance of the proposed selection schemes improves as the power transmit increases. We also observe that RS scheme provides the worst performance and the SBS outperforms all the other selection schemes. It is clear that the EBS scheme converges to the RS asymptotically, i.e., $P_t\rightarrow\infty$, which verifies our discussion in Section III. Due to the fact that the EH model considered has a saturation region, all the selection schemes converge to an error floor, for high values of $P_t$ that validates the corresponding analytical results. We also observe that the IBS scheme converges slowly to the error floor and therefore the convergence cannot be noticed in Fig. \ref{fig_2a}. On the other hand, in Fig. \ref{fig_2b} with $k=4$, we observe that the IBS will eventually converge to the same floor as the SBS scheme asymptotically, i.e., $P_t\rightarrow\infty$, as shown analytically in Section III, but has a slower convergence. The outage performance of the proposed selection schemes based on the linear EH model is used as a benchmark in Fig. \ref{fig_2a}. Finally, theoretical curves match with our simulation results.

\begin{figure}[t]
	\begin{subfigure}{\linewidth}\centering
		\includegraphics[width=\linewidth]{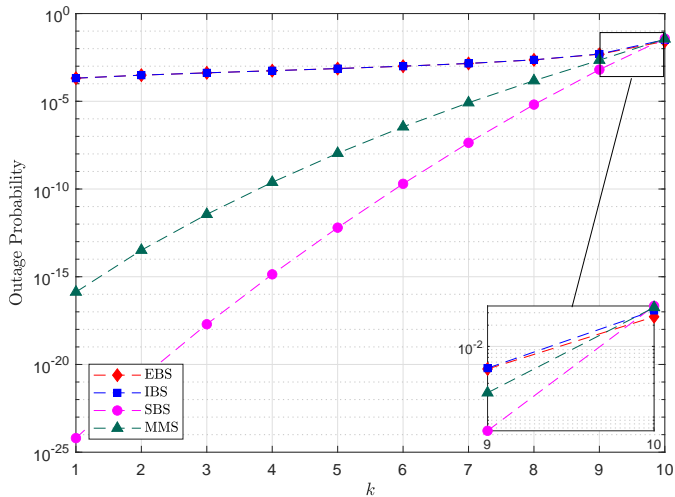}
		\caption{$M=10$.}
	\end{subfigure}	\hfill
	\begin{subfigure}{\linewidth}\centering
		\includegraphics[width=\linewidth]{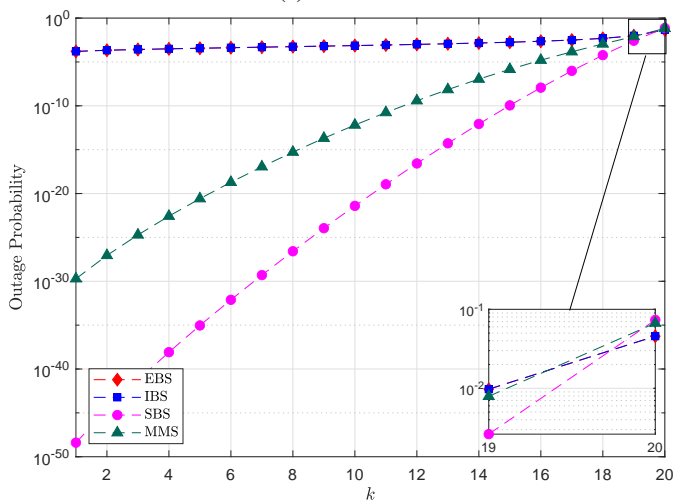}
		\caption{$M=20$.}
	\end{subfigure}	\hfill
	\caption{Outage probability versus $k$; $P_t=-10$ dBm.}
	\label{fig_k1}
\end{figure}

Fig. \ref{fig_k1} highlights the impact of the parameter $k$ on the performance of the system. As the parameter $k$ increases the performance of all the selection schemes degrades, as expected. We also observe that SBS scheme provides the best performance among all the other selection schemes. However, as we discussed before, under a specific scenario, i.e., when $k=M$, EBS scheme outperforms all the other selection schemes and the SBS scheme provides the worst performance. Specifically, the device with the lowest harvested energy might have a stronger uplink channel and due to the saturation effect in the harvesting phase, the downlink channel becomes negligible. Therefore, for this scenario, the uplink transmission dominates the harvesting phase and affects significantly the system's performance.

The effect of the joint selection of the $k$-th and the $l$-th best device on the performance of the system is shown in Fig. \ref{fig_j}. We consider a setup with $k=1, 2$, $j \in \{3, 4,\dots, M\}$ and $M=10, 20, 30$. Due to the restriction that we have pointed out in Section III-F, the threshold level is set as $Q=-4$ dB. The performance improves as the difference of $k$-th and $l$-th increases. For example, the best performance is achieved when the best and the worst device are selected. This follows from the fact that we consider single user detection, and thus the strongest user dominates the weakest and provides the best performance. It can be also noticed that the outage performance improves as the total number of devices increases, as expected.
\begin{figure}[t]\centering 
	\includegraphics[width=\linewidth]{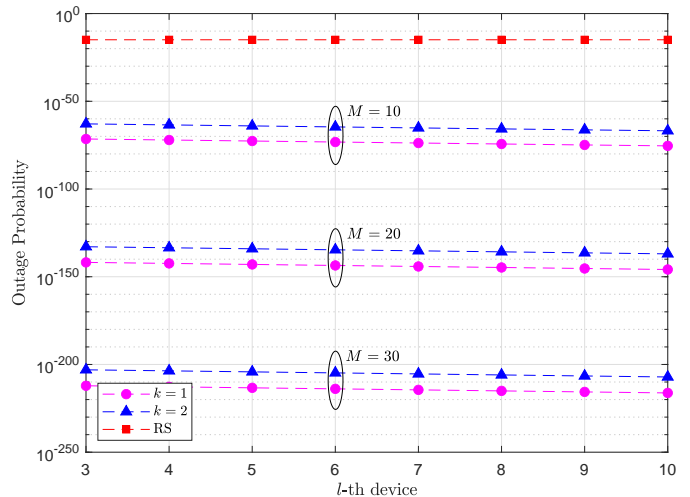}
	\caption{Outage probability of joint selection of $k$-th and $l$-th best device for SBS with $k=1, 2$ and $j=\{3,\dots, M\}$; $P_t=-40$ dBm.}
	\label{fig_j}
\end{figure}
The asymptotic performance of the network which is obtained through EVT for different values of the number of devices is illustrated in Fig. \ref{fig_a}. We observe that the asymptotic outage performance improves as the number of devices increases and the SBS scheme provides the best performance among all the other selection schemes which is in line with the non-asymptotic case. The best selection ($k=1$) is used as a benchmark. Obviously, asymptotically, the harvesting phase is negligible due to saturation, and the communication phase dominates the system's performance. Furthermore, it is important to point out here, that for $P_t \to \infty$, the IBS converges to SBS scheme, however, as $M \to \infty$, a remarkable difference between them is observed. Theoretical curves match with our simulation results and validate our analytical framework.

\begin{figure}[t]\centering 
  \includegraphics[width=\linewidth]{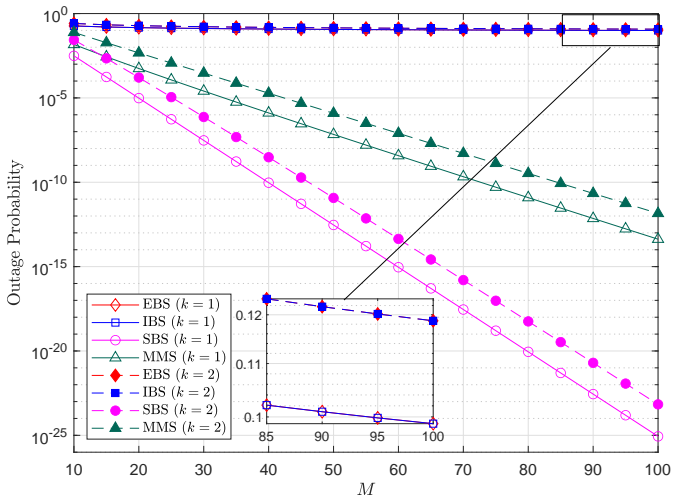}
  \caption{Asymptotic outage probability versus $M$ for $k=1,2$; $P_t=-40$ dBm.}
  \label{fig_a}
\end{figure}

The impact of the harvesting/communication phase duration can be observed in Fig. \ref{fig_t}, which plots the outage probability versus the duration of the harvesting phase $t_1$. A trade-off between the outage probability and $t_1$ can be noticed for all the proposed selection schemes. Specifically, the outage performance improves as $t_1$ increases, as more energy can be harvested by the devices. However, after a certain point the duration of communication phase $t_2$ becomes small which degrades the performance. We provide the point that minimizes the outage probability for each selection scheme through numerical tools such as Matlab \emph{fmin()} function. The optimization of $\frac{t_2}{t_1}$ is critical for the performance of the proposed selection schemes. What is interesting is that the value of $t^*=0.5256$, in this case, is the same for all the proposed selection schemes. Moreover, although for small or high values of $t_1$, all the selection schemes perform similarly, at $t^*$ we observe a remarkable performance gain achieved by the SBS scheme over all the other schemes. In Fig. \ref{fig_t}, we also study the impact of channel estimation error based on the minimum mean square error (MMSE) estimator, assuming that $\sigma_E^2$ is given a priori \cite{MMSE}. Therefore, the estimated channels are complex Gaussian distributed with zero mean and variance $1-\sigma^2_E$, i.e., $\hat{g_i}, \hat{h_i} \sim CN(0,1-\sigma^2_E)$. We observe that as $\sigma_E^2$ increases, the outage performance degrades, as expected.

\begin{figure}[t]\centering 
	\includegraphics[width=\linewidth]{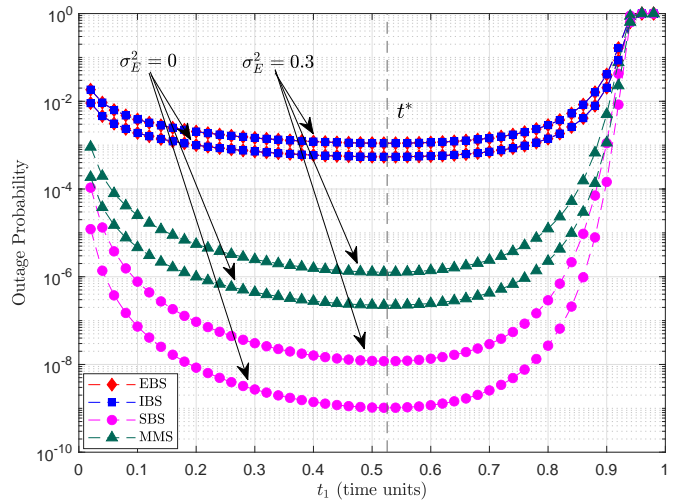}
	\caption{Outage probability versus $t_1$ for $k=2$; $P_t=-10$ dBm.}
	\label{fig_t}
\end{figure}

\section{Conclusions}
In this paper, we have studied the problem of a generalized selection in WPCN, where batteryless devices harvest RF energy in order to communicate with a common receiver. The selection mechanism consists of several novel selection/scheduling schemes corresponding to different implementation complexities and CSI requirements. A complete analytical framework for the performance of the $k$-th best device for both single and pair device selection, was presented. In particular, we considered a non-linear EH model and derived analytical closed-form expressions for the outage probability of the proposed selection schemes, by using tools from high order statistics. Moreover, we considered an asymptotic scenario in terms of the number of devices and by employing EVT, the system's performance was evaluated. %Generally, the SBS scheme provides the best performance among all the other selection schemes for both the asymptotic and non-asymptotic case. Nevertheless, each selection scheme is beneficial for a specific scenario. The SBS scheme was generalized in a scenario where the $k$-th and $l$-th best device are selected for information transmission with single user detection, and we study the impact of joint selection on the outage performance of the system. 
The derived analytical framework provides useful insights for the design of such networks in terms of the main system parameters.

\appendix
\subsection{Proof of Theorem \ref{theorem1}}\label{prf_theorem1}
The CDF of the $i$-th device's SNR $X_i$ is evaluated as
\begin{align}
F_{X_i}(x)&=\mathbb P\left\{t_1\left(\frac{aP_t|g_i|^2+b}{P_t|g_i|^2+c}-\frac{b}{c}\right)\frac{|h_i|^2}{t_2\sigma_n^2}\leq x\right\}\nonumber\\
&=\mathbb P\left\{|h_i|^2 \leq \frac{\sigma_n^2ct_2x(P_t|g_i|^2+c)}{P_tt_1|g_i|^2(ac-b)}\right\}\nonumber\\
&=\mathbb{E}\left[1-\exp\left(-\frac{\sigma_n^2ct_2x(P_t|g_i|^2+c)}{P_tt_1|g_i|^2(ac-b)}\right)\right],\label{e1}
\end{align}
which follows from the CDF of $|h_{i}|^2$ which is an exponential random variable with unit variance. Note that the CDF based on RS can be derived by solving the expression with respect either to $|h_{i}|^2$ or $|g_{i}|^2$.
Thus, we have
\begin{align} \label{random_cdf}
F_{X_i}(x)\!=\!\!\int_{0}^{\infty}\!\!\!\!\! \exp\left(-y\right)\left(1-\exp\left(-\frac{\sigma_n^2ct_2x(P_ty+c)}{P_tt_1y(ac-b)}\right)\right)dy.
\end{align}
The final expression is derived by using \cite[3.324-1]{integrals}, which completes the proof.

\subsection{Proof of Theorem \ref{theorem2}}\label{prf_theorem2}
In order to derive the CDF of the $k$-th best device's SNR for the SBS scheme, we integrate the PDF given by \eqref{pdf}, as follows
\begin{align}
F_{X_{i^*}}(x) &= \int_0^x f_{X_{i^*}}(y)dy\nonumber\\
&= k \binom{M}{k}\!\!\int_0^x \!\!\!f_{X_i}(y)F_{X_i}(y)^{M-k}(1-F_{X_i}(y))^{k-1}dy,
\end{align}
which can be written as \eqref{op_sbs_alouini}, by using the transformation $t \to F_{X_i}(y)$ and the definition of the normalized incomplete beta function \cite{integrals}.

For the linear EH model, we have
\begin{equation}
F_{X_i,\rm{L}}(x)=\int_{0}^{\infty} \exp\left(-y\right)\left(1-\exp\left(-\frac{\sigma_n^2t_2x}{P_tt_1y}\right)\right)dy,
\end{equation} 
which simplifies to \eqref{linear_random} by using \cite[3.324-1]{integrals}.

\subsection{Proof of Theorem \ref{theorem3}}\label{prf_theorem3}
Regarding the EBS scheme, we focus on the harvesting phase and derive the PDF for the $k$-th best device's channel gain $|g_{i^*}|^2$. By combining \eqref{e1} and \eqref{pdf}, the outage probability achieved by the EBS scheme can be expressed as
\begin{align}
\Pi^{(k)}_{\rm EBS}(x)&\!=\!\int_{0}^{\infty}\!\!\left(1-\exp\left(-\frac{\sigma_n^2ct_2x}{t_1(ac-b)}-\frac{\sigma_n^2c^2t_2x}{P_tt_1y(ac-b)}\right)\right)\nonumber \\
&\qquad \qquad \times f_{|g_{i^*}|^2}(y) dy, \label{k-thenergybased}
\end{align}
where $f_{g_{i^*}}(y)$ denotes the PDF of the $k$-th best channel gain $|g_{i^*}|^2$, where $|g_i|^2$ follows an exponential distribution. After some algebraic manipulations, we have
\begin{align}
\Pi^{(k)}_{\rm EBS}(x)&=k\binom{M}{k}\int_{0}^{\infty}\bigg(\exp(-y)^k(1-\exp(-y))^{M-k}\nonumber \\
&\qquad \qquad-\exp(-y)^k(1-\exp(-y))^{M-k}\nonumber \\ 
&\qquad \times \!\exp\!\left(-\frac{\sigma_n^2ct_2x}{t_1(ac-b)}-\frac{\sigma_n^2c^2t_2x}{P_tt_1y(ac-b)}\right)\bigg)dy\nonumber\\
&=k \binom{M}{k}\bigg(\int_{0}^{\infty}\bigg(\frac{(1-\exp(-y))^{M-k}}{\exp(y)^k}dy\nonumber \\
&-\exp\left(-\frac{\sigma_n^2ct_2x}{t_1(ac-b)}\right) \int_{0}^{\infty}\sum_{m=0}^{M-k}(-1)^m\binom{M-k}{m}\nonumber \\
&\times \exp\left(-(k+m)y-\frac{\sigma_n^2c^2t_2x}{P_tt_1y(ac-b)}\right)dy\bigg),
\end{align}
where the first term follows by \cite[3.312.1]{integrals}, the second term follows by \cite[3.324-1]{integrals} and the binomial theorem $(x+y)^n=\sum_{m=0}^{n} \binom{n}{m}x^{n-m}y^m$. Thus, the outage probability of the EBS scheme is given as follows
\begin{align}
\Pi^{(k)}_{\rm EBS}(x)&=k\binom{M}{k}\Bigg(B(k,M-k+1)\nonumber \\
&\qquad-2\sum_{m=0}^{M-k}(-1)^m\binom{M-k}{m} \exp\left(-\frac{\sigma_n^2ct_2x}{t_1(ac-b)}\right)\nonumber\\
&\qquad \qquad \times \sqrt{\frac{\sigma_n^2c^2t_2x}{P_tt_1(ac-b)(k+m)}}\nonumber \\
&\qquad \qquad \times  K_1\left(2\sqrt{\frac{\sigma_n^2c^2t_2x}{P_tt_1(ac-b)}(k+m)}\right)\Bigg),
\end{align}
which can be written as \eqref{op_ebs} with $\Phi(x,m)$ given by \eqref{nonlinear_ebs}. For the linear EH model, the outage probability for the EBS scheme can be expressed as 
\begin{equation} \label{linear_proof}
\Pi^{(k)}_{\rm EBS,L}(x)=\int_{0}^{\infty}\left(1-\exp\left(-\frac{\sigma_n^2t_2x}{P_tt_1y}\right)\right)f_{|g_{i^*}|^2}(y) dy,
\end{equation}
which simplifies to \eqref{op_ebs} with $\Phi_L(x,m)$ defined by \eqref{linear_ebs}, following similar analytical steps with above, with a difference on the CDF of the $i$-th device's SNR $X_i$.

\subsection{Proof of Theorem \ref{theorem4}}\label{prf_theorem4}
The outage probability of the IBS scheme follows similar steps as the analysis of the EBS scheme, presented in Appendix C. However, here, we focus on the communication phase and firstly evaluate the CDF of the $i$-th device's SNR $X_i$ by solving with respect to $|g_{i}|^2$, as
\begin{align}
F_{X_i}(x)=\mathbb{E}\left[1-\exp\left(-\frac{\sigma_n^2c^2t_2x}{P_t(t_1|h_{i}|^2(ac-b)-\sigma_n^2ct_2x)}\right)\right],
\end{align}
which follows from the CDF of $|g_{i}|^2$, which is an exponential random variable. Thus, we have
\begin{align} \label{random_IBS}
F_{X_i}(x)&\!=\!1\!-\!\int_{\frac{\sigma_n^2ct_2x}{t_1(ac-b)}}^{\infty}\!\!\exp\left(-\frac{\sigma_n^2c^2t_2x}{P_t(t_1z(ac-b)-\sigma_n^2ct_2x)}\right) \nonumber \\ & \qquad \qquad \times \exp(-z) dz.
\end{align}
In addition, by using \eqref{pdf}, the outage probability of the IBS scheme can be expressed as follows
\begin{align}
\Pi^{(k)}_{\rm IBS}(x)&=\int_0^{\infty}f_{|h_{i^*}|^2}(z) dz\nonumber \\
&-\int_{\frac{\sigma_n^2ct_2x}{t_1(ac-b)}}^{\infty} \exp\left(-\frac{\sigma_n^2c^2t_2x}{P_t(t_1z(ac-b)-\sigma_n^2ct_2x)}\right)\nonumber \\
&\qquad \qquad \qquad \times f_{|h_{i^*}|^2}(z) dz, \label{k-thinformationbased}
\end{align}
where $f_{|h_{i^*}|^2}(z)$ denotes the PDF of the $k$-th best channel gain $|h_{i^*}|^2$ as $|h_i|^2$ follows an exponential distribution. After some algebraic manipulations, we have
\begin{align}
\Pi^{(k)}_{\rm IBS}(x) &\!=\! k \binom{M}{k}\bigg(\int_0^\infty \exp(-z)^k(1-\exp(-z))^{M-k} dz\nonumber\\
&\qquad \qquad-\int_{\frac{\sigma_n^2ct_2x}{t_1(ac-b)}}^{\infty}\!\!\!\exp(-z)^k(1-\exp(-z))^{M-k}\nonumber\\ &\qquad \qquad \times \exp\!\left(\!-\!\frac{\sigma_n^2c^2t_2x}{P_t(t_1z(ac-b)-\sigma_n^2ct_2x)}\right)dz\!\bigg)\\
&= k \binom{M}{k} \bigg(\int_0^\infty \frac{(1-\exp(-z))^{M-k}}{\exp(z)^k}dz\nonumber \\
&- \int_{\frac{\sigma_n^2ct_2x}{t_1(ac-b)}}^{\infty}\!\!\sum_{m=0}^{M-k}\!\!(-1)^m\binom{M-k}{m}\!\exp(-(k+m)z)\nonumber\\
&\qquad \qquad \times \!\exp\left(\!-\!\frac{\sigma_n^2c^2t_2x}{P_t(t_1z(ac-b)-\sigma_n^2ct_2x)}\right)\!dz\!\bigg),
\end{align}
where the first term follows by \cite[3.312.1]{integrals} and the second term follows by the binomial theorem $(x+y)^n=\sum_{m=0}^{n} \binom{n}{m}x^{n-m}y^m$. Therefore, the outage probability achieved by the IBS scheme is given by
\begin{align}
\Pi^{(k)}_{\rm IBS}(x) &\!=\! k\binom{M}{k}\!\bigg(\!B(k,M-k+1)\!-\!\!\!\sum_{m=0}^{M-k}(-1)^m\binom{\!M-k\!}{\!m\!}\nonumber\\
&\qquad \qquad \times \int_{\frac{\sigma_n^2ct_2x}{t_1(ac-b)}}^{\infty}\exp\bigg(-(k+m)z\nonumber \\
&\qquad \qquad-\frac{\sigma_n^2c^2t_2x}{P_t(t_1z(ac-b)-\sigma_n^2ct_2x)}\bigg)dz\bigg).
\end{align}
For the linear EH model, we follow similar analytical steps with above, with a difference on the CDF of the $i$-th device's SNR $X_i$. Therefore, the outage probability for the IBS scheme with linear EH can be expressed as 
\begin{align} \label{linear_ibs_proof}
\Pi^{(k)}_{\rm IBS,L}(x)=\int_{0}^{\infty}\left(1-\exp\left(-\frac{\sigma_n^2t_2x}{P_tt_1z}\right)\right)f_{|h_{i^*}|^2}(z) dz,
\end{align}
which has the same form with \eqref{linear_proof} and hence, simplifies to \eqref{op_ebs} with $\Phi_L(x,m)$ defined by \eqref{linear_ebs}.

\subsection{Proof of Theorem \ref{theorem5}}\label{prf_theorem5}
\noindent According to the MMS scheme, the outage probability can be expressed as the sum of the outage probabilities conditioned on $|g_{i^*}|^2<|h_{i^*}|^2$ or $|g_{i^*}|^2>|h_{i^*}|^2$, since the worst link of each device pair is determined and then the device pair with the $k$-th strongest worst link is selected. Therefore, we can write
\begin{multline}
\!\!\!\!\!\!\Pi^{(k)}_{\rm MMS}(x)\!\!=\!\!\mathbb P\{X\!\!\leq\! x, |g_{i^*}|^2\!\!<\!\!|h_{i^*}|^2\!\}+\mathbb P\{X\!\leq \!x, |g_{i^*}|^2\!\!>\!\!|h_{i^*}|^2\!\}\\
\!=\frac{1}{2}\mathbb P \bigg\{\!X\!\!\leq \!x\bigg| |g_{i^*}|^2\!\!<\!\!|h_{i^*}|^2\!\bigg\}+\frac{1}{2}\mathbb P\bigg\{\!X\!\! \leq \!x\bigg| |g_{i^*}|^2\!\!>\!\!|h_{i^*}|^2\!\bigg\}.
\end{multline}
By using \eqref{pdf}, the outage probability achieved by the MMS scheme can be expressed as follows
\begin{multline}
\Pi^{(k)}_{\rm MMS}(x)=\underbrace{\frac{1}{2}\int_{y=0}^{s}\!\int_{z=y}^{w}\!\!f_{|h_{i}|^2}(z)f_{|g_{i^*}|^2}(y) dzdy}_\text{$|g_{i^*}|^2<|h_{i^*}|^2$}\\
\!+\underbrace{\frac{1}{2}\left(\int_{z=0}^{r}\!\!f_{|h_{i^*}|^2}(z)dz+\int_{z=r}^{s}\!\int_{y=z}^{v}\!\!\!\!f_{|g_{i}|^2}(y)f_{|h_{i^*}|^2}(z) dydz\right)}_\text{$|g_{i^*}|^2>|h_{i^*}|^2$},\label{mms}
\end{multline}
where $f_{|g_{i}|^2}(y)$ and $f_{|h_{i}|^2}(z)$ are the PDFs of the minimum channel gain $|g_{i}|^2$ and $|h_{i}|^2$, respectively, and $r=\frac{\sigma_n^2ct_2x}{t_1(ac-b)}$, $v=-\frac{cr}{P_t(r-z)}$, $w=r+\frac{cr}{P_ty}$, $s=\sqrt{\frac{r^2}{4}+\frac{cr}{P_t}}+\frac{r}{2}$. By using $f_{|g_{i}|^2}(y) = \exp(-y)$ and $f_{|h_{i}|^2}(z) = \exp(-z)$, the result follows.

For the linear EH model, the outage probability of the MMS scheme can be expressed as in \eqref{mms} with 
\begin{align} \label{parameters}
v=w=s=\sqrt{\frac{\sigma_n^2t_2x}{t_1 P_t}}, r=0,
\end{align}
and by using the binomial theorem $(x+y)^n=\sum_{m=0}^{n} \binom{n}{m}x^{n-m}y^m$, we obtain \eqref{linear_mms}.

%\subsection{Proof of Proposition \ref{proposition_1}}\label{prf_proposition1}

\subsection{Proof of Proposition \ref{proposition_2}}\label{prf_proposition2}
By substituting the CDF of the $i$-th device's channel gain denoted by $F(y)=1-\exp(-y)$ in \eqref{constant_eta} and \eqref{constant_ksi}, we evaluate and substitute $\eta^{\rm EBS}=\log(M)$ and $\xi^{\rm EBS}=1$ in \eqref{op_evt}. Then, by differentiating \eqref{op_evt} and substituting the result in \eqref{k-thenergybased}, we obtain the asymptotic outage probability for the $k$-th best device achieved by the EBS scheme as
\begin{align}
\Pi^{(k)}_{\rm EBS}(x)&=\int_0^\infty\frac{dG^{(k)}\left(y-\log(M)\right)}{dy}\nonumber \\
& \times \left(1-\exp\left(-\frac{\sigma_n^2ct_2x}{t_1(ac-b)}-\frac{\sigma_n^2c^2t_2x}{P_tyt_1(ac-b)}\right)\right)\!dy,
\end{align}
which, after some algebraic manipulations, can be written as
\begin{align}
	\Pi^{(k)}_{\rm EBS}(x)&=1-\frac{M^k}{\Gamma(k)}\exp(-r) \nonumber \\
	&\qquad \times \int_0^\infty\exp\left(-M\exp(-y)-ky-\frac{r c}{P_ty}\right)dy. 
	\end{align}
By using a Taylor series expansion and with the help of \cite[3.324-1]{integrals}, the final result can be obtained.

%For the linear EH model, we substitute $\eta^{\rm EBS}=\log(M)$ and $\xi^{\rm EBS}=1$ in \eqref{op_evt} and by differentiating %\eqref{op_evt} and substituting the result in \eqref{linear_proof}, we obtain the asymptotic outage probability of the $k$-th %best for the EBS scheme as
%\begin{equation}
%\Pi^{(k)}_{\rm EBS,L}(x) = %\int_0^\infty\frac{dG^{(k)}\left(y-\log(M)\right)}{dy}\left(1-\exp\left(-\frac{\sigma_n^2t_2x}{P_tt_1y}\right)\right)dy.
%\end{equation}
%which simplifies to \eqref{ebs_evt_linear}.

\subsection{Proof of Proposition \ref{proposition_3}}\label{prf_proposition3}
Similarly to the EBS scheme, we use the CDF of the $i$-th device's channel gain, i.e. $F(z)=1-\exp(-z)$, in \eqref{constant_eta} and \eqref{constant_ksi}. Then, we evaluate and substitute $\eta^{\rm IBS}=\log(M)$ and $\xi^{\rm IBS}=1$ in \eqref{op_evt}. Finally, by differentiating \eqref{op_evt} and substituting the result in \eqref{k-thinformationbased}, we obtain the asymptotic outage probability of the $k$-th best for the IBS scheme as
\begin{align}
\Pi^{(k)}_{\rm IBS}(x)&=\int_\frac{\sigma_n^2ct_2x}{t_1(ac-b)}^\infty\frac{dG^{(k)}\left(z-\log(M)\right)}{dz}\nonumber \\
&\quad \times \left(1-\exp\left(-\frac{\sigma_n^2c^2t_2x}{P_t(t_1z(ac-b)-\sigma_n^2ct_2x)}\right)\right)dz,
\end{align}
which, after some algebraic manipulations, can be written as
\begin{align}
	\Pi^{(k)}_{\rm IBS}(x) &= \frac{M^k}{\Gamma(k)} \int_r^\infty \left(1-\exp\left(-\frac{c r}{P_t(z-r)}\right)\right)\nonumber \\
	&\qquad \qquad \times \exp(-M\exp(-z)-kz)dz,
\end{align}
and the proposition is proved.
%For the linear EH model, we substitute $\eta^{\rm IBS}=\log(M)$ and $\xi^{\rm IBS}=1$ in \eqref{op_evt} and by differentiating %\eqref{op_evt} and substituting the result in \eqref{linear_ibs_proof}, we obtain the asymptotic outage probability of the %$k$-th best for the IBS scheme as
%\begin{equation}
%\Pi^{(k)}_{\rm %IBS,L}(x)=\int_0^\infty\frac{dG^{(k)}\left(z-\log(M)\right)}{dz}\left(1-\exp\left(-\frac{\sigma_n^2t_2x}{P_tt_1z}\right)\right)dz.
%\end{equation}

\subsection{Proof of Proposition \ref{proposition_4}}\label{prf_proposition4}
The EVT is applied over the worst links from each device pair. By assuming that the CDF of the $i$-th device's worst channel gain is denoted by $F(y)=1-\exp(-2y)$ and $F(z)=1-\exp(-2z)$, for the cases $|g_{i^*}|^2<|h_{i^*}|^2$ and $|g_{i^*}|^2>|h_{i^*}|^2$, respectively, we calculate the normalizing constants as $\eta^{\rm MMS}=\frac{1}{2}\log(M)$ and $\xi^{\rm MMS}=\frac{1}{2}$. Then, by differentiating \eqref{op_evt} and substituting the result in \eqref{mms}, we find the asymptotic outage probability of the $k$-th best for the MMS scheme as
\begin{align} \label{mms_evt}
\Pi^{(k)}_{\rm MMS}(x) &= \frac{1}{2}\int_0^s \int_y^w f_{|h_{i}|^2}(z) \frac{dG^{(k)}\left(2y-\log(M)\right)}{dy} dz dy\nonumber\\
&+\frac{1}{2}\bigg(\int_0^r f_{|h_{i^*}|^2}(z)dz \nonumber \\
&\quad+ \int_r^s \int_z^v f_{|g_{i}|^2}(y) \frac{dG^{(k)}\left(2z-\log(M)\right)}{dz} dy dz\bigg)\\
&=\frac{M^k}{\Gamma(k)}\!\! \int_0^s\!\! \int_y^w \!\!\exp(-z-\!M\exp(-2y)\!-\!2ky) dz dy\nonumber\\
&+ k \binom{M}{k} \int_0^r \exp(-2kz)(1-\exp(-2z))^{M-k}dz \nonumber\\
&+ \frac{M^k}{\Gamma(k)} \int_r^s \!\!\int_z^v \!\!\exp(-y-\!M\exp(-2z)\!-\!2kz) dy dz.
\end{align}
By utilizing the binomial theorem and by using a Taylor series expansion, the asymptotic outage probability for the MMS scheme can be written as
%\begin{multline}
%=\frac{M^k}{\Gamma(k)}\sum_{n=0}^{\infty}\frac{1}{n!}\sum_{m=0}^{\infty}\binom{n}{m}(-M)^{n-m}\bigg((-2k\!-\!1)^m(2(n-m))^{-m-1}\gamma(m+1,2(n-m)s)\\+\left((-2k-1)^m\!\!-\exp(-v)(-2k)^m\right)\left((2(n\!-\!m))^{-1-m}(\gamma(m+1,2(n-m)r)\!-\!\gamma(m\!+\!1,2(n\!-\!m)s))\right)\!\!\bigg)\\
%+k\binom{M}{k}\sum_{m=0}^{M-k}(-1)^m\frac{1-\!\exp(-2(k+m)r)}{2(k+m)},
%\end{multline}
\begin{align}
\Pi^{(k)}_{\rm MMS}(x) &= \frac{1}{\Gamma(k)}\sum_{n=0}^{\infty}(-1)^n\frac{1}{n!}\sum_{m=0}^{\infty}\binom{n}{m}M^{k+n-m}\nonumber\\ 
& \times \!\!(2n-2m)^{-m-1} \!\bigg(\!\!(2k+1)^m\gamma(m+1,2(n-m)s)\\
& +((2k+1)^m-(2k)^m\exp(-v))\nonumber\\
&\times\!(\gamma(m+1,2(n-m)r)\!-\!\gamma(m+1,2(n-m)s))\bigg) \nonumber\\
&+k\binom{M}{k}\sum_{m=0}^{M-k}(-1)^m\frac{1-\exp(-2(k+m)r)}{2(k+m)},
\end{align}
which follows with the help of \cite[3.351-1]{integrals}. Then, we obtain the final result through algebraic operations.

\begin{thebibliography}{10}
\bibitem{ICC} M. Dimitropoulou, C. Psomas and I. Krikidis, ``$k$-th best device selection for scheduling in wireless powered communication networks,'' in \emph{Proc. Int. Conf. Commun. (ICC)}, Dublin, Ireland, 2020, pp. 1-6.
\bibitem{order-statistics-} H. A. David, H. N. Nagaraja, \emph{Order Statistics}. Wiley Series in Probability and Statistics, 2004.
\bibitem{tse} D. Tse and P. Viswanath, \emph{Fundamentals of wireless communication}. Cambridge, U.K.: Cambridge University Press, 2005.
\bibitem{goldsmith} A. Goldsmith, \emph{Wireless Communications}. Cambridge, U.K.: Cambridge University Press, 2005.
\bibitem{WPCN1} S. Bi, Y. Zeng, and R. Zhang, ``Wireless powered communication networks: An overview,'' \emph{IEEE Wireless Commun.}, vol. 23, no. 2, pp. 10-18, Apr. 2016.
\bibitem{survey} X. Lu, P. Wang, D. Niyato, D. I. Kim and Z. Han, ``Wireless charging technologies: Fundamentals, standards, and network applications,'' \emph{IEEE Commun. Surveys Tut.}, vol. 18, no. 2, pp. 1413–1452, 2016.
\bibitem{alouinibook} H. Yang and M. Alouini, Order Statistics in Wireless Communications: Diversity, Adaptation, and Scheduling in MIMO and OFDM Systems. Cambridge: Cambridge University Press, 2011.
\bibitem{coop} Y. Zou, X. Wang and W. Shen, ``Optimal Relay Selection for Physical-Layer Security in Cooperative Wireless Networks,'' \emph{IEEE J. Sel. Areas Commun.}, vol. 31, no. 10, pp. 2099-2111, Oct. 2013.
\bibitem{cognitive} Y. Zou, J. Zhu, B. Zheng, and Y.-D. Yao, ``An adaptive cooperation diversity scheme with best-relay selection in cognitive radio networks,'' \emph{IEEE Trans. Signal Process.}, vol. 58, no. 10, pp. 5438-5445, Oct. 2010.
\bibitem{best_relay} S. Ikki and M. Ahmed, ``Performance analysis of adaptive decode-and-forward cooperative diversity networks with best-relay selection,'' \emph{IEEE Trans. Commun.}, vol. 58, no. 1, pp. 68-72, Jan. 2010.
\bibitem{why k-th}S. S. Ikki and M. H. Ahmed, ``On the performance of cooperative-diversity networks with the Nth best-relay selection scheme,'' \emph{IEEE Trans. Commun.}, vol. 58, no. 11, pp. 3062-3069, Nov. 2010.
\bibitem{why n-th} X. Zhang, Z. Yan, Y. Gao and W. Wang, ``On the study of outage performance for cognitive relay networks (CRN) with the Nth best-relay selection in Rayleigh-fading channels,'' \emph{IEEE Wireless Commun. Lett.}, vol. 2, no. 1, pp. 110-113, Feb. 2013.
\bibitem{alouini3} Y. H. Al-Badarneh, C. N. Georghiades and M. Alouini, ``On the asymptotic throughput of the $k$-th best secondary user selection in cognitive radio systems,'' in \emph{Proc. IEEE Veh. Technol. Conf. (VTC-Fall)}, Chicago, IL, Aug. 2018.
\bibitem{EVT} L. Haan, and A. Ferreira, ``Extreme value theory: an introduction,'' Springer Series in Operations Research and Financial Engineering. Springer, New York, 2006.
\bibitem {evt} Q. Wang and Y. Jing, ``Closed-form average SNR and ergodic capacity approximations for best relay selection,'' \emph{IEEE Trans. Veh. Technol.}, vol. 65, no. 4, pp. 2827-2833, Apr. 2016.
\bibitem{alouini2} Y. H. Al-Badarneh, C. N. Georghiades and M. Alouini, ``Asymptotic performance analysis of the $k$th best link selection over wireless fading channels: An extreme value theory approach,'' \emph{IEEE Trans. Veh. Technol.}, vol. 67, no. 7, pp. 6652-6657, Jul. 2018.
\bibitem{alouini4} Y. H. Al-Badarneh, C. N. Georghiades and M. Alouini, ``Asymptotic performance analysis of generalized user selection for interference-limited multiuser secondary networks,'' \emph{IEEE Trans. Cognitive Commun. Netw.}, vol. 5, pp. 82-92, Mar. 2019.
\bibitem{WPCN2} H. Ju and R. Zhang, ``Throughput maximization in wireless powered communication networks,'' \emph{IEEE Trans. Wireless Commun.}, vol. 13, no. 1, pp. 418–428, Jan. 2014.
\bibitem{book_wpcn} S. Nikoletseas, Y. Yang, and A. Georgiadis, \emph{Wireless Power Transfer Algorithms, Technologies and Applications in Ad Hoc Communication Networks}. Springer, 2016.
\bibitem{SWIPT1} R. Zhang and C. K. Ho, ``MIMO broadcasting for simultaneous wireless information and power transfer,'' \emph{IEEE Trans. Wireless Commun.}, vol. 12, no. 5, pp. 1989-2001, May 2013.
\bibitem{SWIPT2} G. Pan, H. Lei, Y. Yuan, and Z. Ding, ``Performance analysis and optimization for SWIPT wireless sensor networks,'' \emph{IEEE Trans. Commun.,}, vol. 65, no. 5, pp. 2291–2302, May 2017.
\bibitem{diomidis}  D. Michalopoulos, H. Suraweera, and R. Schober, ``Relay selection for simultaneous information transmission and wireless energy transfer: A tradeoff perspective,'' \emph{IEEE J. Sel. Areas Commun.}, vol. 33, pp. 1578-1594, Aug. 2015.
\bibitem{morsi1} R. Morsi, D. Michalopoulos, and R. Schober, ``Multi-user scheduling schemes for simultaneous wireless information and power transfer over fading channels,'' \emph{IEEE Trans. Wireless Commun.}, vol. 14, no. 4, pp. 1967–1982, Apr. 2015.
\bibitem{morsi2} E. Boshkovska, R. Morsi, D. W. K. Ng, and R. Schober, ``Power allocation and scheduling for SWIPT systems with non-linear energy harvesting model,'' in \emph{Proc. Inter. Conf. on Commun. (ICC)}, Kuala Lumpur, Malaysia, May 2016.
\bibitem{kim} I. Bang, S. M. Kim, and D. K. Sung, ``Adaptive multiuser scheduling for simultaneous wireless information and power transfer in a multicell environment,'' \emph{IEEE Trans. Wireless Commun.}, vol. 16, no. 11, pp. 7460–7474, Nov. 2017.
%\bibitem {linear} B. Medepally and N. B. Mehta, ``Voluntary energy harvesting relays and selection in cooperative wireless networks,'' \emph{IEEE Trans. Wireless Commun.}, vol. 9, no. 11, pp. 3543-3553, Nov. 2010.
\bibitem{clerckx} B. Clerckx, R. Zhang, R. Schober, D. W. K. Ng, D. I. Kim and H. V. Poor, ``Fundamentals of Wireless Information and Power Transfer: From RF Energy Harvester Models to Signal and System Designs,'' \emph{IEEE J. Sel. Areas Commun.}, vol. 37, no. 1, pp. 4-33, Jan. 2019.
\bibitem {bruno} M. Varasteh, J. Hoydis, B. Clerckx, ``Learning modulation design for SWIPT with nonlinear energy harvester: Large and small signal power regimes,'' in \emph{Proc. Inter. Conf. Wireless Commun. Signal Process.}, Cannes, France, Jul. 2019.
\bibitem{alouini1} Y. Chen, N. Zhao and M. Alouini, ``Wireless energy harvesting using signals from multiple fading channels,'' \emph{IEEE Trans. Commun.}, vol. 65, no. 11, pp. 5027-5039, Nov. 2017.
\bibitem {schober} E. Boshkovska, D. W. K. Ng, N. Zlatanov, and R. Schober, ``Practical non-linear energy harvesting model and resource allocation for SWIPT systems,'' \emph{IEEE Commun. Lett.}, vol. 19, no. 12, pp. 2082-2085, Dec. 2015.
\bibitem {non linear} Y. Chen, K. S. Thomas, and R. A. Abd-Alhameed, ``New formula for conversion efficiency of RF EH and its wireless applications,'' \emph{IEEE Trans. Veh. Technol.}, vol. 65, no. 11, pp. 9410-9414, Nov. 2016.
\bibitem {green radio} C. Han et al., ``Green radio: radio techniques to enable energy-efficient wireless networks,'' \emph{IEEE Commun. Mag.}, vol. 49, pp. 46-54, June 2011.
\bibitem {consumption} H. Liang, C. Zhong, X. Chen, H. A. Suraweera, and Z. Zhang, ``Wireless powered dual-hop multi-antenna relaying systems: Impact of CSI and antenna correlation,'' \emph{IEEE Trans. Wireless Commun.}, vol. 16, no. 4,	pp. 2505-2519, Apr. 2017.


%\bibitem{proakis} J. Proakis and M. Salehitm, Digital Communications, Fifth ed. New York, NY, USA: McGraw-Hill, 2007.

%\bibitem{ibrahim} A.S Ibrahim, A.K. Sadek, W. Su and K.J.R Liu. ``Cooperative communications with relay-selection: When to %cooperate and whom to cooperate with?,'' \emph{IEEE Trans. Wireless Commun.}, vol. 7, no. 7, pp. 2814-2817, Jul. 2008.
%\bibitem{ding} Z. Ding, H. Dai and H. V. Poor, ``Relay selection for cooperative NOMA,'' \emph{IEEE Wireless Commun. Lett.}, vol. 5, no. 4, pp. 416-419, Aug. 2016.
%\bibitem{harvard} S. Nam, M. Vu and V. Tarokh, ``Relay selection methods for wireless cooperative communications,'' in \emph{Proc. Annual Conf. on Information Sciences and Systems}, Princeton, NJ, pp. 859-864, Mar. 2008.
%\bibitem {e2e} A. Ikhlef and M. Z. Bocus, ``Outage performance analysis of relay selection in SWIPT systems,'' in \emph{Proc. IEEE Wireless Commun. Netw. Conf.}, Doha, Qatar, Apr. 2016, pp. 1-5.
%\bibitem{energy} J. Chen, C. Liu and M. Qian, ``A selection-based cooperative SWIPT scheme with energy-preserving DF relays,'' in \emph{Proc. Inter. Conf. Wireless Commun. Signal Process.}, Hangzhou, China, Oct. 2018, pp. 1-6.
%\bibitem {krikidis} I. Krikidis, ``Relay selection in wireless powered cooperative networks with energy storage,'' \emph{IEEE J. Sel. Areas Commun.}, vol. 33, no. 12, pp. 2596-2610, Dec. 2015.
%\bibitem {access} Q. Li, S. Feng, A. Pandharipande, X. Ge, Q. Ni and J. Zhang, ``Wireless-powered cooperative multi-relay systems with relay selection,'' \emph{IEEE Access}, vol. 5, pp. 19058-19071, Sep. 2017.
%\bibitem{best} T. Mekkawy, R. Yao, N. Qi, and Y. Lu, ``Secure relay selection for two way amplify-and-forward untrusted relaying networks,'' \emph{IEEE Trans. Veh. Technol.}, vol. 67, no. 12, pp. 11979-11987, Dec. 2018.


%\bibitem{nth} X. Zhang, Y. Zhang, Z. Yan, J. Xing and W. Wang, ``Performance analysis of cognitive relay networks over nakagami- $m$ fading channels,'' \emph{IEEE Sel. Areas Commun.}, vol. 33, no. 5, pp. 865-877, May 2015.



\bibitem{bessel_small} Z.-H. Yang and Y.-M. Chu, ``On approximating the modified Bessel function of the second kind,'' \emph{Journal of Inequalities and Applications}, vol. 41, Dec. 2017.
\bibitem{maxmin1} A. Bletsas, A. Khisti, D. P. Reed and A. Lippman, ``A simple cooperative diversity method based on network path selection,'' \emph{IEEE J. Sel. Areas Commun.}, vol. 24, no. 3, pp. 659-672, Mar. 2006.
\bibitem{walsh} I. E. Walsh, ``Sample sizes for approximate independence of largest and smallest order statistics,'' \emph{Journal of the American Statistical Association}, vol. 65, no. 330, pp. 860-863, Jun. 1970.
\bibitem{dbm} X. Zhou, R. Zhang, and C. K. Ho, ``Wireless information and power transfer: Architecture design and rate-energy tradeoff,'' \emph{IEEE Trans.	Commun.}, vol. 61, no. 11, pp. 4754-4767, Nov. 2013.
\bibitem{MMSE} T. Yoo and A. Goldsmith, ``Capacity and power allocation for fading MIMO channels with channel estimation error,'' \emph{IEEE Trans. Inf. Theory}, vol. 52, pp. 2203-2214, May 2006.
\bibitem{integrals} I. S. Gradshteyn and I. M. Ryzhik, \emph{Table of integrals, series, and products}, 7-th Ed., Elsevier Ac. Press, 2007.
\end{thebibliography}
\end{document}